\documentclass[twocolumn,amsthm]{autart}

\usepackage{amsmath,amsfonts,amssymb}

\usepackage{algorithmic}

\usepackage{array}
\usepackage[
    caption=false,
    font=normalsize,
    labelfont=sf,
    textfont=sf
]{subfig}
\usepackage{textcomp}
\usepackage{stfloats}
\usepackage{url}
\usepackage{verbatim}
\usepackage{graphicx}
\usepackage{comment}
\usepackage{booktabs}
\usepackage{multirow}
\usepackage{enumitem}
\usepackage{newtxtext}
\usepackage{xcolor}
\usepackage{tikz}
\usetikzlibrary{
    arrows.meta,
    positioning,
    fit,
    backgrounds,
    calc
}

\newtheoremstyle{mythmstyle}
  {6pt}
  {6pt}
  {\itshape}
  {}
  {\bfseries}
  {.}
  {0.5em}
  {%
    \thmname{#1} \thmnumber{#2}%
    \if\relax\detokenize{#3}\relax
    \else\ {\normalfont\itshape(#3)}%
    \fi
  }

\theoremstyle{mythmstyle}

\newtheorem{theorem}{Theorem}
\newtheorem{thm}[theorem]{Theorem}

\newtheorem{lemma}{Lemma}
\newtheorem{lem}[lemma]{Lemma}

\newtheorem{proposition}{Proposition}
\newtheorem{prop}[proposition]{Proposition}

\newtheorem{corollary}{Corollary}
\newtheorem{cor}[corollary]{Corollary}

\newtheorem{definition}{Definition}
\newtheorem{defn}[definition]{Definition}

\newtheorem{remark}{Remark}
\newtheorem{rmk}[remark]{Remark}

\newtheorem{claim}{Claim}

\newtheorem{assumption}{Assumption}
\newtheorem{as}[assumption]{Assumption}

\definecolor{NodeBlue}{RGB}{44,123,182}
\definecolor{NeighborOrange}{RGB}{230,126,34}
\definecolor{StableGreen}{RGB}{39,174,96}
\definecolor{AlertRed}{RGB}{192,57,43}
\definecolor{SoftGray}{RGB}{245,245,245}
\definecolor{DarkGray}{RGB}{70,70,70}

\definecolor{ClusterBlue}{RGB}{52,152,219}
\definecolor{ClusterOrange}{RGB}{230,126,34}
\definecolor{ClusterPurple}{RGB}{155,89,182}

\begin{document}
\setlength\abovedisplayskip{0.5pt}
        \setlength\belowdisplayskip{0.5pt}
        \setlength\abovedisplayshortskip{0.5pt}
        \setlength\belowdisplayshortskip{0.5pt}
        \allowdisplaybreaks
        \setlength{\parindent}{1em}
        \setlength{\parskip}{0em}  
        \addtolength{\oddsidemargin}{0.5pt}
\begin{frontmatter}
\title{Cluster-Based Distributed Small-Signal Stability Certificates for Grid-Forming Inverter Networks\thanksref{t1}}

\thanks[t1]{Corresponding Author: Sijia Geng.}

\author[Baiae]{Bhathiya Rathnayake}\ead{brathna1@jh.edu}~and    
\author[troy]{Sijia Geng}\ead{sgeng@jhu.edu}            

\address[Baiae]{Ralph O'Connor Sustainable Energy Institute (ROSEI), Johns Hopkins University, Baltimore, MD, USA}  
\address[troy]{Department of Electrical and Computer Engineering, Johns Hopkins University, Baltimore, MD, USA}

\begin{abstract}
Large-scale power networks are often organized by geography, ownership, or control authority, making stability certificates that require a fully assembled global model challenging. This paper develops a time-domain small-signal stability certification framework for grid-forming inverter networks with selectable clustering resolution. The objective is to certify stability at the same scale at which the network is organized and operated: each cluster verifies conditions using intra-cluster and limited boundary information, and these checks collectively yield a network-level stability certificate without requiring a global eigenvalue computation. After linearization about a phase-cohesive synchronized operating point, a small-angle approximation decomposes the model into a voltage subsystem and an angle-frequency subsystem, where the latter is certified by an energy argument using the symmetric weighted-Laplacian network structure. For the voltage subsystem, node-to-node gains are introduced and a cyclic small-gain argument yields a family of sufficient exponential stability certificates ranging from fully decentralized to cluster-based and centralized. For an arbitrary network partition, each cluster verifies intra-cluster directed-cycle conditions and inter-cluster path conditions. The singleton- and single-cluster limits recover the decentralized and centralized certificates, respectively. The resulting stability indices provide diagnostic information beyond a pass/fail verdict by localizing the limiting margin to individual nodes, internal feedback loops, and inter-cluster channels. The analysis further reveals distinct roles of gain normalization and network partitioning. Proportional normalization minimizes the decentralized index at each node over all admissible normalization weights, whereas, for any fixed normalization, the cluster-based certificate can improve upon the decentralized certificate only if the selected partition internalizes all channels responsible for the decentralized failure. The framework is validated on a synthetic network and the IEEE 39-bus system through comparison with eigenvalue-based benchmarks across multiple clustering resolutions and normalizations.
\end{abstract}

\end{frontmatter}

\section{Introduction}

\begin{figure*}
    \centering
\includegraphics[width=0.85\linewidth]{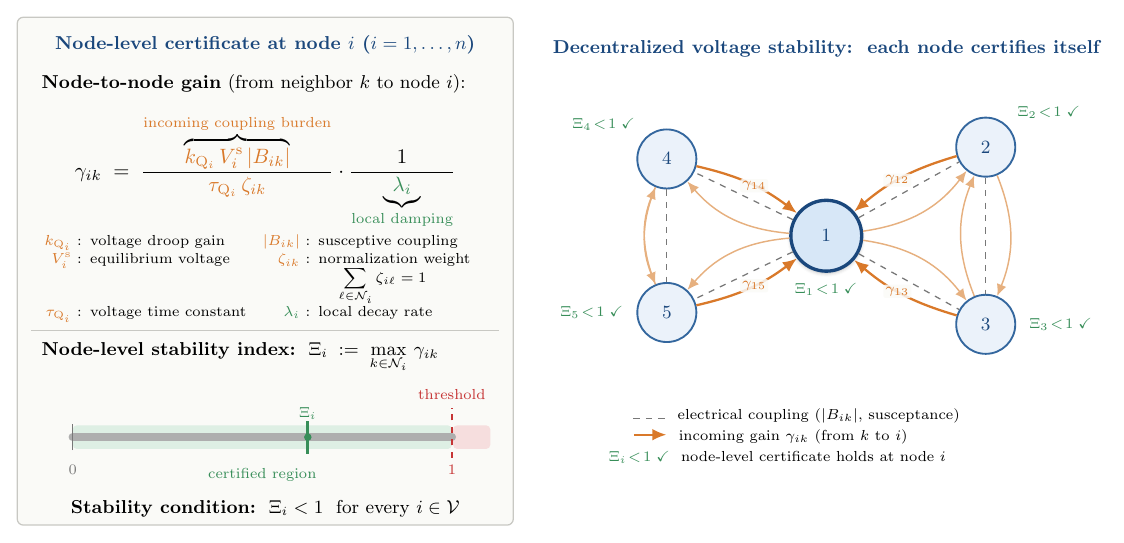}
    \caption{Decentralized voltage stability certification. Left: the node-to-node
gain $\gamma_{ik}$ and the node-level stability index $\Xi_i$. 
Right: a five-node example in which every node verifies its own condition
$\Xi_i<1$; the incoming gains $\gamma_{1k}$, $k\in\mathcal N_1$, are highlighted for node $1$.}
    \label{fig:nodeFig}
\end{figure*}

\begin{figure*}
    \centering
\includegraphics[width=0.85\linewidth]{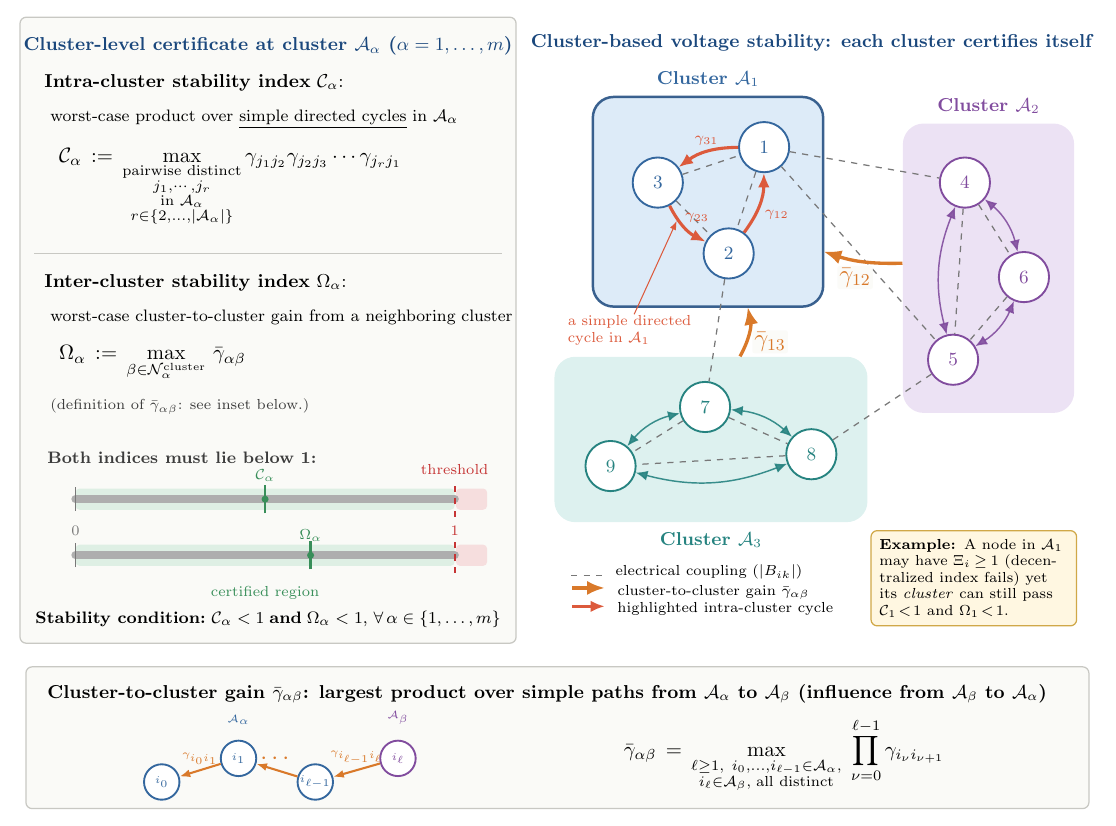}
    \caption{Cluster-based voltage stability certification. Left: each cluster
$\mathcal A_\alpha$ is certified by two indices, the intra-cluster index
$\mathcal C_\alpha$ and the inter-cluster index $\Omega_\alpha$; both must lie
below one. Right: a nine-node network partitioned into three clusters, with
one simple directed cycle in $\mathcal A_1$ highlighted and the incoming
cluster-to-cluster gains $\bar\gamma_{12}$ and $\bar\gamma_{13}$ directed from
$\mathcal A_2$ and $\mathcal A_3$, respectively, to $\mathcal A_1$. Bottom: $\bar\gamma_{\alpha\beta}$ is the largest gain product over simple
paths from $\mathcal A_\alpha$ to $\mathcal A_\beta$; under the
receiver--source convention, the corresponding influence propagates in the
reverse direction, from $\mathcal A_\beta$ to $\mathcal A_\alpha$.}
\label{fig:cluster}
\end{figure*}

The increasing adoption of inverter-based resources (IBRs) is changing the mechanisms that govern power system stability \cite{chatterjee2025voltage,chatterjee2025effects,chatterjee2024sensitivity}. Grid-forming (GFM) inverters are designed to impose a voltage waveform and provide frequency and voltage support, thereby providing functions that were historically associated with synchronous machines \cite{geng2022unified}. In classical machine-dominated systems, small-signal studies have typically been built around centralized models of generator, exciter, governor, and network dynamics \cite{kundur1994power,gibbard2015small}. Although these devices are heterogeneous, their models have long been incorporated into established system-level study workflows \cite{kundur1994power}. In inverter-dominated networks, by contrast, the stability margin can depend sensitively on inverter-level control dynamics \cite{milano2018foundations} while the corresponding device models and control loops are often available only as proprietary black boxes \cite{sun2023stability}. Small-signal stability certification is therefore not merely a matter of computing eigenvalues of a linearized model. It also depends on what information is available for certification, the scale at which the certificate can be verified, and how the resulting margin should be interpreted when no single certification party has access to the full model.

A substantial literature has addressed this limited information challenge through decentralized certificates that rely on local information. By analyzing input admittance across distinct frequency regions, \cite{vorobev2019decentralized} derives per-component stability rules using dissipativity and multiplier theories for inverter-based microgrids. A scale-free frequency control framework using decentralized $\mathcal H_\infty$ synthesis guarantees robust stability independently of the network topology or operating point \cite{pates2019robust}. For grid-connected inverters, small-gain stability certificates bound inverter admittances via grounded network spectral quantities \cite{huang2020h}. Gain-and-phase conditions refine this frequency-domain viewpoint by combining magnitude and phase margins \cite{huang2024gain}. Geometric frequency-domain methods use scaled relative graph analysis to characterize small-signal stability boundaries across varying inverter operating points \cite{baron2026impact}, as well as projected Davis–Wielandt (DW) shells \cite{huang2025geometric} and unified DW-shell separations \cite{feng2025unified}, to reduce the conservativeness of classical decentralized criteria such as the small-gain, small-phase, and passivity theorems. Complementary state-space approaches include port-Hamiltonian and incremental passivity certificates \cite{farokhian2019stability}, distributed Lyapunov stability certificates with dynamic-phasor models \cite{nandanoori2020distributed}, equilibrium-passivity for hybrid AC/DC networks \cite{watson2020control}, and output-differential-passivity certificates for heterogeneous nonlinear bus dynamics \cite{yang2019distributed}. Closest to the current work are the decentralized certificates based on loop transformations \cite{siahaan2024decentralized,wang2026decentralized} that address non-Laplacian network challenges, which were later extended in to the dynamic-line case \cite{haberle2026decentralized}, and small-phase certificates formulated in the complex-frequency domain \cite{niehues2025small}. They provide operating-point-dependent stability conditions and clarify how network coupling terms shape small-signal stability.

This paper addresses a further question: at what resolution should
stability certification be performed? The aforementioned works demonstrate that stability certification can be performed under limited information, but their
conditions are primarily formulated at the device or component level rather
than treating the certification resolution itself as a design choice. Large inverter fleets are not organized solely as electrical networks; they are also divided across geographical regions, plant owners, and control authorities \cite{pudjianto2007virtual}. In such settings, the relevant certification unit may be neither a single inverter nor the entire system, but an intermediate cluster within which device and coupling information can be shared, while information exchange across cluster boundaries remains limited. This motivates a stability certificate whose resolution can be selected to
match the scale at which the network is organized and operated.

We develop such a certificate for droop-controlled GFM inverter networks. Following \cite{schiffer2014conditions,siahaan2024decentralized,wang2026decentralized,pates2019robust}, we assume that the network is lossless and inductive. Around a phase-cohesive synchronized operating point, the linearized model is written in reduced angle coordinates. Under the standard small-angle decoupling assumption \cite{siahaan2024decentralized,wang2026decentralized,simpson2016voltage,jafarian2018interconnection,VanCutsem1998Voltage}, the model separates into an angle-frequency subsystem and a voltage subsystem. The angle-frequency subsystem has a symmetric weighted-Laplacian structure under reciprocal susceptive couplings, and its exponential stability is established by an energy argument. The voltage subsystem consists of locally damped voltage dynamics interconnected through neighbor couplings, making it amenable to a cyclic small-gain analysis based on MAX-preserving maps \cite{karafyllis2011vector}.

The key idea is to express voltage-error propagation through products of node-to-node gains. Each gain measures the influence of a neighboring voltage deviation relative to the local damping of the receiving inverter and can
therefore be interpreted as a coupling-to-damping ratio; see
Fig. \ref{fig:nodeFig}. The gains also depend on normalization weights that allocate the incoming coupling burden across the incident channels of each
receiving node. We consider a power-law family that includes uniform and proportional normalization as special cases.

A cyclic small-gain certificate based on these node-to-node gains admits a natural decomposition under network partitioning. For an arbitrary partition into pairwise disjoint nonempty clusters, internal feedback is characterized
through products of gains around simple directed cycles contained within each cluster, while incoming cross-boundary influence is characterized through
products of gains along admissible inter-cluster paths; see Fig. \ref{fig:cluster}. The corresponding cluster-to-cluster gain is the largest such path product. This yields two
conditions for each cluster: an intra-cluster cycle condition and an inter-cluster path condition. If every cluster satisfies both, the
voltage subsystem is exponentially stable.

The resulting certificate is separable at the selected clustering resolution. Once the node-to-node gains are computed, each cluster verifies its intra-cluster conditions using gains internal to the cluster and its inter-cluster conditions using those intra-cluster gains together with the boundary gains connecting it to neighboring clusters. No global voltage state matrix or information about non-neighboring clusters is required. The same construction applies to every partition: the singleton partition yields a fully decentralized node-level certificate, while
the single-cluster partition yields a centralized cyclic small-gain certificate. Intermediate partitions therefore provide a certification resolution aligned with the physical or organizational structure of the
network.

A further contribution is to distinguish the roles of gain normalization and network partitioning. At the decentralized node level, proportional normalization minimizes the worst incoming gain over all admissible
normalization weights and thus minimizes the node-level stability index. The resulting condition is equivalent to strict row diagonal dominance of the
voltage state matrix. Under the additional no-shunt assumption, the
decentralized condition of \cite{siahaan2024decentralized} is recovered as a sufficient specialization of this proportional certificate.

Partitioning affects conservatism through a different mechanism. For a fixed normalization, the decentralized certificate implies the cluster-based
certificate for every partition, whereas the converse need not hold. However, the cluster-based certificate can improve upon a decentralized failure only if the selected partition internalizes all channels whose gains are responsible for that failure. How readily this can be achieved depends on how the selected normalization distributes the limiting gains among the neighbors of a failing node. 

The framework also produces interpretable stability indices. The node-level index identifies the largest incoming coupling-to-damping ratio at an
individual inverter; see Fig. \ref{fig:nodeFig}. The intra-cluster index identifies the largest gain product around an internal feedback loop, while the
inter-cluster index identifies the strongest incoming influence from a neighboring cluster; see Fig. \ref{fig:cluster}. Relative to their unit threshold,
these indices localize the limiting certificate margin to an individual node, an internal cycle, or an inter-cluster channel, thereby providing diagnostic information beyond a binary spectral verdict, which could inform control retuning.
Numerical studies on a synthetic inverter network and the IEEE 39-bus system benchmark the proposed certificates against eigenvalue-based ground truth.

The remainder of the paper is organized as follows. Section \ref{sec:modeling} introduces the droop-controlled grid-forming inverter and network model. Section \ref{sec:voltage} presents the cluster-based voltage stability certificate and the associated indices. Section \ref{sec:angle} establishes the stability of the angle-frequency subsystem. Section \ref{sec:simulation} presents the numerical study, and Section \ref{sec:conclusions} concludes the paper. The Appendix provides the proof of the cluster-based voltage stability. 

\section{Modeling and Linearization}\label{sec:modeling}

\subsection{Network and Node Model}
\textit{Notation:} The symbols $\mathbb R$, $\mathbb R_{>0}$, $\mathbb R_{+}$, $\mathbb C$, and $\mathbb N$ denote the real numbers, positive real numbers, nonnegative real numbers, complex numbers, and natural numbers, respectively. For a finite set $\mathcal S$, $|\mathcal S|$ denotes its cardinality, while $|x|$ denotes the absolute value of a scalar $x$. The notation $(\cdot)^\top$ and $\|\cdot\|_2$ denote transpose and the vector $2$-norm, respectively. For a symmetric matrix $\boldsymbol A$, $\boldsymbol A\succ0$ and $\boldsymbol A\succeq0$ denote positive definiteness and positive semidefiniteness. The operators $\operatorname{col}(\cdot)$ and $\operatorname{diag}(\cdot)$ denote column stacking and diagonal-matrix construction, respectively; $\boldsymbol 1_n$, $\boldsymbol 0_n$, and $\boldsymbol I_n$ denote the all-ones vector, zero vector, and identity matrix of dimension $n$. We adopt the convention that the maximum over an empty set is zero. Gain subscripts follow a receiver--source convention: $\gamma_{ik}$ denotes the gain from node $k$ to node $i$, and $\bar{\gamma}_{\alpha\beta}$ denotes the gain from cluster $\mathcal A_\beta$ to cluster $\mathcal A_\alpha$. The gain
product
$\gamma_{i_0i_1}\cdots\gamma_{i_{\ell-1}i_\ell}$ is associated with the path represented by the ordered node sequence $(i_0,i_1,\ldots,i_\ell)$, which is said to start at the receiving node
$i_0$ and end at the source node $i_\ell$. The corresponding physical
influence propagates in the reverse order, $i_\ell\to i_{\ell-1}\to\cdots\to i_0$. When $i_\ell=i_0$, the same convention represents a directed cycle.

We consider a network of GFM inverter units\footnote{When algebraic load buses represented by constant impedances  are present, the corresponding load-bus variables can be eliminated by Kron reduction \cite{kundur1994power}, yielding an equivalent network whose nodes correspond only to the GFM inverter buses.}.
Let
\begin{align*}
\mathcal V := \{1,\ldots,n\}, \qquad n\geq 2,
\end{align*}
denote the set of inverter nodes in the inverter network. The electrical interconnection among the nodes is characterized by the complex admittances
\begin{align*}
Y_{ik}:=G_{ik}+\mathrm{i}B_{ik}\in\mathbb C,
\qquad i,k\in\mathcal V,
\end{align*}
where $\mathrm{i}:=\sqrt{-1}$ and $G_{ik},B_{ik}\in\mathbb R$ denote, respectively, the conductance and susceptance. For $i\neq k$, $Y_{ik}$ represents the electrical coupling admittance between inverter nodes $i$ and $k$. For $i\neq k$, if $Y_{ik}=0$, then nodes $i$ and $k$ are not adjacent. Further, we assume that  
$Y_{ik}\neq 0$ if and only if $Y_{ki}\neq 0$ for all $i\neq k$. Thus, the network is the undirected graph
\begin{align*}
\mathcal G := (\mathcal V,\mathcal E),
\end{align*}
with edge set
\begin{align*}
\mathcal E :=
\bigl\{\{i,k\}\subset \mathcal V : i\neq k,\; Y_{ik}\neq 0\bigr\}.
\end{align*}
The neighbor set of node $i$ is therefore
\begin{align*}
\mathcal N_i :=
\{k\in\mathcal V\setminus\{i\}: \{i,k\}\in\mathcal E\}.
\end{align*}

For each node $i$, we associate the voltage phase angle
$\phi_i\in [0, 2\pi)$, the frequency $\omega_i\in\mathbb R$, and the voltage magnitude
$V_i\geq 0$. For each edge $\{i,k\}\in\mathcal E$, define the relative phase
angle
\begin{align*}
\phi_{ik} := \phi_i-\phi_k.
\end{align*}

The net active and reactive power injections at node $i$ are
\begin{align}
P_i
&\!=\!
G_{ii} V_i^2
\!-\!
\sum_{k \in \mathcal N_i}
\!V_i V_k
\bigl(
G_{ik}\!\cos(\phi_{ik})
\!+\!
B_{ik}\!\sin(\phi_{ik})
\bigr),
\label{eq:Pi_general}
\\
Q_i
&\!=\!
\!-\!
B_{ii} V_i^2
\!-\!\!\!
\sum_{k \in \mathcal N_i}
\!V_i V_k
\bigl(\!
G_{ik}\!\sin(\phi_{ik})
-
B_{ik}\!\cos(\phi_{ik})
\!\bigr),
\label{eq:Qi_general}
\end{align}
where the diagonal terms are defined as
\begin{align*}
G_{ii}
&:=
\hat G_{ii} + \sum_{k \in \mathcal N_i} G_{ik},
&
B_{ii}
&:=
\hat B_{ii} + \sum_{k \in \mathcal N_i} B_{ik}.
\end{align*}
Here, $\hat G_{ii}$ and $\hat B_{ii}$ represent shunt conductance and susceptance at
node $i$, respectively.

\begin{as}[Lossless inductive network \cite{schiffer2014conditions}]
\label{ass_lossless}
For all $i,k\in\mathcal V$ with $i\neq k$,
\begin{align*}
G_{ik}=0,
\qquad
B_{ik}\leq 0.
\end{align*}
Moreover, the shunt conductances and susceptances satisfy
\begin{align*}
\hat G_{ii}=0,
\qquad
\hat B_{ii}\leq 0,
\qquad i\in\mathcal V .
\end{align*}
\end{as}

\begin{as}[Connected susceptance graph]
\label{as:connected_reduced_graph}
Under Assumption \ref{ass_lossless}, the susceptance graph
$\mathcal G=(\mathcal V,\mathcal E)$
is connected. Equivalently, for any two distinct nodes $i,k\in\mathcal V$, there
exist an integer $r\geq 1$ and a sequence of nodes $
i=i_0,i_1,\ldots,i_r=k$ such that
$
B_{i_{\ell-1}i_\ell}\neq 0,\,
\ell=1,\ldots,r.$
\end{as}

Since $B_{ii}=\hat B_{ii}+\sum_{k\in\mathcal N_i}B_{ik}$,
Assumption \ref{ass_lossless} implies $B_{ii}\leq 0$ for all
$i\in\mathcal V$. Under Assumption \ref{ass_lossless}, the power
expressions \eqref{eq:Pi_general},\eqref{eq:Qi_general} reduce to
\begin{align}
P_i
&=
\sum_{k \in \mathcal N_i}
|B_{ik}| V_i V_k \sin(\phi_{ik}),
\label{eq:Pi_lossless}
\\
Q_i
&=
|B_{ii}| V_i^2
-
\sum_{k \in \mathcal N_i}
|B_{ik}| V_i V_k \cos(\phi_{ik}).
\label{eq:Qi_lossless}
\end{align} \hfill

Each node is modeled as a droop-controlled GFM inverter:
\begin{align*}
\dot{\phi}_i
&=
\omega_i,
\\
\tau_{\mathrm{P}_i}\dot{\omega}_i
&=
-
\omega_i
+
\omega^{\rm d}
-
k_{\mathrm{P}_i}\bigl(P_i - P_i^{\rm d}\bigr),
\\
\tau_{\mathrm{Q}_i}\dot{V}_i
&=
-
V_i
+
V_i^{\rm d}
-
k_{\mathrm{Q}_i}\bigl(Q_i - Q_i^{\rm d}\bigr),
\end{align*}
where $\tau_{\mathrm{P}_i}>0$ and $\tau_{\mathrm{Q}_i}>0$ are the active-power and
reactive-power filter time constants, respectively, while
$k_{\mathrm{P}_i}>0$ and $k_{\mathrm{Q}_i}>0$ are the active-power/frequency and
reactive-power/voltage droop gains. The scalar $\omega^{\rm d}>0$ denotes the
nominal frequency setpoint, $V_i^{\rm d}>0$ denotes the voltage magnitude
setpoint, and $P_i^{\rm d},Q_i^{\rm d}\in\mathbb R$ denote the desired active
and reactive power injections at node $i$.

In compact form,
\begin{align}
\dot{\boldsymbol \phi}
&=
\boldsymbol \omega,
\label{eq:compact_delta}
\\
\boldsymbol T_{\rm P} \dot{\boldsymbol \omega}
&=
-
\boldsymbol \omega
+
\boldsymbol 1_n \omega^{\rm d}
-
\boldsymbol K_{\rm P} (\boldsymbol P - \boldsymbol P^{\rm d}),
\label{eq:compact_omega}
\\
\boldsymbol T_{\rm Q} \dot{\boldsymbol V}
&=
-
\boldsymbol V
+
\boldsymbol V^{\rm d}
-
\boldsymbol K_{\rm Q} (\boldsymbol Q - \boldsymbol Q^{\rm d}),
\label{eq:compact_voltage}
\end{align}
where $\boldsymbol \phi:=\operatorname{col}(\phi_i)\in \mathbb R^n$, $\boldsymbol \omega:=\operatorname{col}(\omega_i)\in \mathbb R^n$, $\boldsymbol V:=\operatorname{col}(V_i)\in \mathbb R^n$, $\boldsymbol P
:=
\operatorname{col}(P_i) \in \mathbb R^n$, $\boldsymbol P^{\rm d}
:=
\operatorname{col}(P_i^{\rm d}) \in \mathbb R^n$, $\boldsymbol Q
:=
\operatorname{col}(Q_i) \in \mathbb R^n$, $\boldsymbol Q^{\rm d}
:=
\operatorname{col}(Q_i^{\rm d}) \in \mathbb R^n$, $\boldsymbol V^{\rm d}
:=
\operatorname{col}(V_i^{\rm d}) \in \mathbb R^n$, $\boldsymbol T_{\rm P}
:=
\operatorname{diag}(\tau_{\mathrm{P}_i}) \in \mathbb R^{n \times n}$, $\boldsymbol T_{\rm Q}
:=
\operatorname{diag}(\tau_{\mathrm{Q}_i}) \in \mathbb R^{n \times n}$, $\boldsymbol K_{\rm P}
:=
\operatorname{diag}(k_{\mathrm{P}_i}) \in \mathbb R^{n \times n},$ $\boldsymbol K_{\rm Q}
:=
\operatorname{diag}(k_{\mathrm{Q}_i}) \in \mathbb R^{n \times n}.$

\begin{as}[Existence of a phase-cohesive synchronized operating point
\cite{schiffer2014conditions}]
\label{ass_sync_operating_point}
Let the admissible phase-cohesive set be defined by
\begin{align}
\Phi
\!:=\!
\left\{
\boldsymbol\phi\in [0,2\pi)^n:
|\phi_i-\phi_k|<\frac{\pi}{2},
\,
\forall \{i,k\}\in\mathcal E
\right\}.
\label{eq:phase_condition}
\end{align}
There exist a phase vector $\boldsymbol\phi^{\rm s}\in\Phi$, a synchronous
frequency $\omega^{\rm s}\in\mathbb R$, and a positive voltage vector
$\boldsymbol V^{\rm s}\in\mathbb R_{>0}^n$ such that
\begin{align*}
\boldsymbol 1_n \omega^{\rm s}
-
\boldsymbol 1_n \omega^{\rm d}
+
\boldsymbol K_{\rm P}
\bigl(
\boldsymbol P(\boldsymbol\phi^{\rm s},\boldsymbol V^{\rm s})-\boldsymbol P^{\rm d}
\bigr)
&=
\boldsymbol 0_n,
\\
\boldsymbol V^{\rm s}
-
\boldsymbol V^{\rm d}
+
\boldsymbol K_{\rm Q}
\bigl(
\boldsymbol Q(\boldsymbol\phi^{\rm s},\boldsymbol V^{\rm s})-\boldsymbol Q^{\rm d}
\bigr)
&=
\boldsymbol 0_n.
\end{align*}
\end{as}

Thus, the corresponding synchronized motion of \eqref{eq:compact_delta}-\eqref{eq:compact_voltage} is 
\begin{align}
\boldsymbol \phi^\star(t)
&=
\operatorname{mod}_{2\pi}
\bigl(
\boldsymbol \phi^{\rm s} + \boldsymbol 1_n \omega^{\rm s} t
\bigr),
\label{eq:sync_motion_delta}
\\
\boldsymbol \omega^\star(t)
&=
\boldsymbol 1_n \omega^{\rm s},
\label{eq:sync_motion_omega}
\\
\boldsymbol V^\star(t)
&=
\boldsymbol V^{\rm s}.
\label{eq:sync_motion_voltage}
\end{align}
Since only angle differences affect the power-flow expressions \eqref{eq:Pi_lossless},\eqref{eq:Qi_lossless}, the synchronized
motion can be studied in a rotating reference frame \cite{schiffer2014conditions}. Define the shifted variables
\begin{align*}
\boldsymbol{\tilde \omega}(t)
&:=
\boldsymbol \omega(t) - \boldsymbol 1_{n} \omega^{\rm s},
\qquad
\boldsymbol{\tilde \phi}(t)
&:=
\boldsymbol \phi(0) + \int_0^t \boldsymbol{\tilde \omega}(\tau)\,d\tau.
\end{align*}
To eliminate the rotational degree of freedom, select node $n$ as the reference and define
\begin{align*}
\boldsymbol \theta
:=
\boldsymbol{\mathcal R}\boldsymbol{\tilde \phi},\qquad\qquad
\boldsymbol{\mathcal R}
:=
\begin{bmatrix}
\boldsymbol I_{n-1} & -\boldsymbol 1_{n-1}
\end{bmatrix}.
\end{align*}
Thus $\boldsymbol \theta = (\theta_1,\ldots,\theta_{n-1})^\top \in \mathbb R^{n-1}$, and we set
\begin{align*}
\theta_n := 0.
\end{align*}
It follows that
\begin{align*}
\theta_i = \phi_i - \phi_n,
\qquad
\forall i\in\mathcal V\setminus\{n\}.
\end{align*}

Under Assumption \ref{ass_sync_operating_point}, the synchronized motion
\eqref{eq:sync_motion_delta}-\eqref{eq:sync_motion_voltage} corresponds, in the
rotating frame, to the equilibrium
\begin{align}
(\boldsymbol\theta,\tilde{\boldsymbol\omega},\boldsymbol V)
=
(\boldsymbol\theta^{\rm s},\boldsymbol 0_n,\boldsymbol V^{\rm s}),
\label{eq:reduced_equilibrium}
\end{align}
where $\boldsymbol\theta^{\rm s}
:=
\boldsymbol{\mathcal R}\boldsymbol\phi^{\rm s}.
$
Moreover, with the convention $\theta_n^{\rm s}:=0$,
\begin{align*}
\theta_{ik}^{\rm s}:=\theta_i^{\rm s}-\theta_k^{\rm s}
=
\phi_i^{\rm s}-\phi_k^{\rm s}
=
\phi_{ik}^{\rm s},
\end{align*}
for all $\{i,k\}\in\mathcal E$. 

In the reduced coordinates $(\boldsymbol\theta,\tilde{\boldsymbol\omega},\boldsymbol V)$,
the dynamics take the following form. Let 
\begin{align*}
   \mathcal{V}_0 := \mathcal  V\setminus\{n\}. 
\end{align*}
For each $i \in \mathcal V_0$,
\begin{align}
\dot{\theta}_i
&=
\tilde \omega_i - \tilde \omega_n,
\label{eq:reduced_theta_i}
\\
\begin{split}
\tau_{\mathrm{P}_i}\dot{\tilde \omega}_i
&\!=\!
-
\!\tilde \omega_i
\!-\!
k_{\mathrm{P}_i}
\!\!\sum_{k \in \mathcal N_i}
|B_{ik}| V_i V_k \sin(\theta_{ik})
+
\omega^{\rm d} - \omega^{\rm s} \\&\quad+ k_{\mathrm{P}_i}P_i^{\mathrm{d}},
\label{eq:reduced_omega_i}
\end{split}
\\
\begin{split}
\tau_{\mathrm{Q}_i}\dot V_i
&\!=\!
-
V_i
\!-\!
k_{\mathrm{Q}_i}
\Big(\!
|B_{ii}|V_i^2
\!-\!\!\!
\sum_{k \in \mathcal N_i}
\!\!|B_{ik}|V_iV_k\!\cos(\theta_{ik})
\!\Big)
\\&\quad+
V_i^{\rm d}+ k_{\mathrm{Q}_i}Q_i^{\rm d},
\label{eq:reduced_V_i}
\end{split}
\end{align}
where
$\theta_{ik} := \theta_i - \theta_k.$

For the reference node $n$,  since $\theta_n=0$, we have
\begin{align}
\begin{split}
\tau_{\mathrm{P}_n}\dot{\tilde \omega}_n
&\!=\!
-
\tilde \omega_n
\!+\!
k_{\mathrm{P}_n}
\!\!\!\sum_{k \in \mathcal N_n}\!\!
|B_{nk}|V_nV_k\!\sin(\theta_k)
\!+\!
\omega^{\rm d} \!-\! \omega^{\rm s} \\&\quad+ k_{\mathrm{P}_n}P_n^{\mathrm d},
\label{eq:reduced_omega_n}
\end{split}
\\
\begin{split}
\tau_{\mathrm{Q}_n}\!\!\dot V_n
&\!\!=\!\!
-
\!V_n
\!-\!
k_{\mathrm{Q}_n}
\!\!\Big(\!
|B_{nn}|V_n^2
\!-\!\!\!\!\!
\sum_{k \in \mathcal N_n}
\!\!\!|B_{nk}|V_nV_k\!\cos(\theta_k)
\!\!\Big)
\\&\quad +
V_n^{\rm d} + k_{\mathrm{Q}_n}Q_n^{\rm d}.
\label{eq:reduced_V_n}
\end{split}
\end{align}

\subsection{Linearization around the synchronized operating point}

Recall that $(\boldsymbol \theta^{\rm s},\boldsymbol 0_n,\boldsymbol V^{\rm s})$ is the equilibrium corresponding to the synchronized solution \eqref{eq:sync_motion_delta}-\eqref{eq:sync_motion_voltage} in reduced coordinates. Introduce the deviations
\begin{align*}
\delta \theta_i
&:=
\theta_i - \theta_i^{\rm s},
\qquad
\forall i \in \mathcal V_0,
\\
\delta \omega_i
&:=
\tilde \omega_i,
\qquad
\forall i \in \mathcal V,
\\
\delta V_i
&:=
V_i - V_i^{\rm s},
\qquad
\forall i \in \mathcal V.
\end{align*}
Since $\theta_n \equiv 0$, for any edge involving the reference node one has $\theta_{kn}^{\rm s} = \theta_k^{\rm s}$. Moreover,
\begin{align*}
\delta \theta_{ik}
:=
\theta_{ik} - \theta_{ik}^{\rm s}
=
\delta \theta_i - \delta \theta_k.
\end{align*}
The linearized dynamics for $i \in \mathcal V_0$ are
\begin{align}
\delta \dot \theta_i
&=
\delta \omega_i -\delta \omega_n,
\label{eq:lin_theta_i}
\\
\begin{split}
\delta \dot \omega_i
&=
-
\frac{1}{\tau_{\mathrm{P}_i}}\delta \omega_i
\!-\!
\frac{k_{\mathrm{P}_i}}{\tau_{\mathrm{P}_i}}
\!\!\sum_{k \in \mathcal N_i}
|B_{ik}|
\Bigl(
V_i^{\rm s} V_k^{\rm s} \cos(\theta_{ik}^{\rm s})(\delta \theta_i - \delta \theta_k)
\\&\qquad\qquad+
V_k^{\rm s} \sin(\theta_{ik}^{\rm s})\delta V_i
+
V_i^{\rm s} \sin(\theta_{ik}^{\rm s})\delta V_k
\Bigr),\label{eq:lin_omega_i}
\end{split}
\\
\begin{split}
\delta \dot V_i
&=
-
\frac{1}{\tau_{\mathrm{Q}_i}}\delta V_i
-
\frac{k_{\mathrm{Q}_i}}{\tau_{\mathrm{Q}_i}}
\Biggl(
2|B_{ii}|V_i^{\rm s} \delta V_i
\\&\quad-
\sum_{k \in \mathcal N_i}
|B_{ik}|
\Bigl(
-
V_i^{\rm s} V_k^{\rm s} \sin(\theta_{ik}^{\rm s})(\delta \theta_i - \delta \theta_k)
\\&\qquad+
V_k^{\rm s} \cos(\theta_{ik}^{\rm s})\delta V_i
+
V_i^{\rm s} \cos(\theta_{ik}^{\rm s})\delta V_k
\Bigr)
\Biggr).
\label{eq:lin_V_i}
\end{split}
\end{align}
When $k=n$, the convention $\delta \theta_n = 0$ is used. 

For the reference node,
\begin{align}
\begin{split}
\delta \dot \omega_n
&\!=\!
-
\frac{1}{\tau_{\mathrm{P}_n}}\delta \omega_n
\!+\!
\frac{k_{\mathrm{P}_n}}{\tau_{\mathrm{P}_n}}
\!\!\!\sum_{k \in \mathcal N_n}
\!\!|B_{nk}|
\Bigl(
V_n^{\rm s} V_k^{\rm s} \cos(\theta_k^{\rm s})\delta \theta_k
\\&\quad+
V_k^{\rm s} \sin(\theta_k^{\rm s})\delta V_n
+
V_n^{\rm s} \sin(\theta_k^{\rm s})\delta V_k
\Bigr),
\label{eq:lin_omega_n}
\end{split}
\\
\begin{split}
\delta \dot V_n
&=
-
\frac{1}{\tau_{\mathrm{Q}_n}}\delta V_n
-
\frac{k_{\mathrm{Q}_n}}{\tau_{\mathrm{Q}_n}}
\Biggl(
2|B_{nn}|V_n^{\rm s} \delta V_n
\\&\quad-
\sum_{k \in \mathcal N_n}
|B_{nk}|
\Bigl(
-
V_n^{\rm s} V_k^{\rm s} \sin(\theta_k^{\rm s})\delta \theta_k
\\&\qquad\quad+
V_k^{\rm s} \cos(\theta_k^{\rm s})\delta V_n
+
V_n^{\rm s} \cos(\theta_k^{\rm s})\delta V_k
\Bigr)
\Biggr).
\label{eq:lin_V_n}
\end{split}
\end{align}

\subsection{Decoupling}

\begin{as}[Small-angle decoupling approximation]
\label{ass:small_angle_decoupling}
Following \cite{siahaan2024decentralized,simpson2016voltage}, the small-angle approximation used in the subsequent analysis is implemented through the substitutions
\begin{align}
\sin(\theta_{ik}^{\rm s}) &\mapsto 0,
&
\cos(\theta_{ik}^{\rm s}) &\mapsto 1,
\qquad
\{i,k\}\in\mathcal E.\label{eq:small_angle_trig_approx}
\end{align}
\end{as}
Applying \eqref{eq:small_angle_trig_approx} to the linearized model
\eqref{eq:lin_theta_i}-\eqref{eq:lin_V_n} eliminates the angle-voltage
cross-coupling terms. The resulting approximate linearized dynamics decouple into
an angle-frequency subsystem and a voltage subsystem.

\noindent \underline{Angle-frequency subsystem}
\begin{align}
\delta \dot \theta_i
&=
\delta \omega_i-\delta \omega_n,
\qquad i\in\mathcal V_0,
\label{eq:af_theta_dyn}
\\
\begin{split}
\delta \dot \omega_i
&=
-\frac{1}{\tau_{\mathrm{P}_i}}\delta \omega_i
-
\frac{k_{\mathrm{P}_i}}{\tau_{\mathrm{P}_i}}
\sum_{k\in\mathcal N_i}
|B_{ik}|V_i^{\rm s}V_k^{\rm s}\delta \theta_i
\\
&\quad+
\frac{k_{\mathrm{P}_i}}{\tau_{\mathrm{P}_i}}
\sum_{k\in\mathcal N_i\setminus\{n\}}
|B_{ik}|V_i^{\rm s}V_k^{\rm s}\delta \theta_k,
\qquad i\in\mathcal V_0,
\end{split}
\label{eq:af_omega_dyn_i}
\\
\delta \dot \omega_n
&=
-\frac{1}{\tau_{\mathrm{P}_n}}\delta \omega_n
+
\frac{k_{\mathrm{P}_n}}{\tau_{\mathrm{P}_n}}
\sum_{k\in\mathcal N_n}
|B_{nk}|V_n^{\rm s}V_k^{\rm s}\delta \theta_k.
\label{eq:af_omega_dyn_n}
\end{align}

\noindent \underline{Voltage subsystem}
\begin{align}\label{eq:linearized_volt_dynamics}
\begin{split}
\delta \dot V_i
&=
-
\lambda_i\delta V_i
+
\frac{k_{\mathrm{Q}_i}V_i^{\rm s}}{\tau_{\mathrm{Q}_i}}
\sum_{k \in \mathcal N_i}
|B_{ik}|  \delta V_k,
\end{split}
\end{align}
where for each $i \in \mathcal V$
\begin{align}\label{eq:lambda_i}
   \lambda_i := \frac{1 + k_{\mathrm{Q}_i}D_i}{\tau_{\mathrm{Q}_i}},
\end{align}
with
\begin{align}
D_i
&:=
2|B_{ii}|V_i^{\rm s}
-
\sum_{k \in \mathcal N_i}
|B_{ik}| V_k^{\rm s}.
\label{eq:Di_def}
\end{align}

\section{Cluster-Based Voltage Stability Criterion}\label{sec:voltage}

In this section, we study the voltage subsystem
\eqref{eq:linearized_volt_dynamics}-\eqref{eq:Di_def}, which can be compactly written as
\begin{align*}
    \delta \dot{\boldsymbol V}
    =
    \boldsymbol A_{\rm v}\delta \boldsymbol V,
\end{align*}
where $\delta\boldsymbol V
   :=
   \operatorname{col}(\delta V_i)_{i\in\mathcal V}
   \in \mathbb R^n$
and $\boldsymbol A_{\rm v}\in\mathbb R^{n\times n}$ is the state matrix induced by
\eqref{eq:linearized_volt_dynamics}. Thus, the stability of the voltage subsystem
is characterized by the Hurwitzness of $\boldsymbol A_{\rm v}$. This
provides an exact centralized test, but it does not reveal how stability is distributed
across the network or identify which clusters or coupling channels are responsible for the loss or preservation of stability. Moreover, $\boldsymbol A_{\rm v}$ is Metzler, so when $\lambda_i>0$ for all $i\in\mathcal{V}$ in  \eqref{eq:lambda_i}, the voltage subsystem admits an M-matrix spectral stability characterization \cite{berman1994nonnegative}. Such a test is also centralized and does not directly provide separable per-cluster certificates or identify the limiting internal loops and inter-cluster channels. 

The objective of this section is to address this limitation. Suppose that the
inverter network is partitioned into clusters. Can one certify voltage stability of the
entire network through conditions that are checked cluster by cluster, using only intra-cluster
information and limited boundary information from neighboring clusters? Moreover, can such
conditions be expressed in a form that is interpretable enough to guide corrective action
by network operators?

To this end, we develop a flexible cluster-based distributed stability certificate for the voltage subsystem
\eqref{eq:linearized_volt_dynamics}-\eqref{eq:Di_def}. The proposed certificate separates two mechanisms
through which voltage perturbations may propagate: node-to-node gain cycles circulating perturbations within each
cluster, and cluster-to-cluster gain paths through which perturbations enter one cluster
from another. The resulting conditions are local at the cluster level and quantify whether
each cluster has sufficient internal damping to attenuate both its own internal feedback
loops and the strongest incoming influence from neighboring clusters. As a special case,
the framework also yields a fully decentralized node-level stability condition.

We first define the node-level and cluster-level gains used in the distributed cluster-based stability criterion.

\begin{defn}[Node-to-node and cluster-to-cluster gains]
\label{def:hierarchical_voltage_gains} \rm 
Let $\mathcal V=\{1,\ldots,n\}$ denote the set of grid-forming inverter nodes. Consider the voltage subsystem
\eqref{eq:linearized_volt_dynamics}-\eqref{eq:Di_def}, and suppose that
\begin{align*}
\lambda_i>0,
\qquad i\in\mathcal V,
\end{align*}
where $\lambda_i$ is given by \eqref{eq:lambda_i},\eqref{eq:Di_def}. For each node $i\in\mathcal V$, let $\{\zeta_{ik}\}_{k\in\mathcal N_i}$ be positive normalization weights satisfying
\begin{align}
\zeta_{ik}>0,\qquad k\in\mathcal N_i,\qquad
\sum_{k\in\mathcal N_i}\zeta_{ik}=1.
\label{eq:zeta_allocation}
\end{align} For each ordered pair $(i,k)\in\mathcal V\times\mathcal V$, define the \textit{node-to-node gain} from node $k$ to node $i$ by
\begin{align}
\gamma_{ik}
:=
\begin{cases}
\dfrac{k_{\mathrm{Q}_i} V_i^{\rm s}|B_{ik}|}
{\tau_{\mathrm{Q}_i}\lambda_i\zeta_{ik}}, & k\in\mathcal N_i,\\[3mm]
0, & k\notin\mathcal N_i,
\end{cases}
\label{eq:gammaik_def}
\end{align}
where $k_{\mathrm{Q}_i},\tau_{\mathrm{Q}_i}>0$ for all $i\in\mathcal V$. Thus $\gamma_{ik}\ge 0$ for all $i,k\in\mathcal V$. See Fig. \ref{fig:nodeFig} (top left).  

Let
\begin{align*}
\mathcal V=\mathcal A_1\cup\cdots\cup\mathcal A_m,
\end{align*}
be a partition of $\mathcal V$ into pairwise disjoint nonempty clusters. For each ordered pair of distinct clusters $(\mathcal A_\alpha,\mathcal A_\beta)$, $\alpha,\beta\in\{1,\cdots, m\}$,
define the \textit{cluster-to-cluster gain} from $\mathcal A_\beta$ to $\mathcal A_\alpha$ by
\begin{align}
\begin{split}
\bar\gamma_{\alpha\beta}
\!:=\!
\max
\bigg\{\!
\prod_{\nu=0}^{\ell-1}&\gamma_{i_\nu i_{\nu+1}}
:
\ell\ge 1,\\&\quad 
i_0,\ldots,i_\ell \text{ are pairwise distinct},\\&\quad
i_0,\ldots,i_{\ell-1}\in\mathcal A_\alpha,
i_\ell\in\mathcal A_\beta
\bigg\}.
\end{split}
\label{eq:cluster_gain_def}
\end{align}
Thus, $\bar\gamma_{\alpha\beta}$ is the largest gain product over simple paths from $\mathcal A_\alpha$ to $\mathcal A_\beta$. Under the receiver--source convention, the corresponding influence propagates in the reverse direction, from $\mathcal A_\beta$ to $\mathcal A_\alpha$. See Fig. \ref{fig:cluster} (bottom). 

Finally, define the effective inter-cluster neighbor set
\begin{align}
\mathcal N_\alpha^{\rm cluster}
:=
\big\{
\beta\in\{1,\ldots,m\}\setminus\{\alpha\}:
\bar\gamma_{\alpha\beta}>0
\big\}.
\label{eq:cluster_neighbor_set_def}
\end{align}
\end{defn}


\begin{rmk}[Choices of normalization weights]\label{rmk:zeta_choices}\rm
The normalization weights $\zeta_{ik}$ in Definition
\ref{def:hierarchical_voltage_gains} are not network or controller
parameters. They determine the scaling assigned to each incoming channel in the certificate. A useful one-parameter family is
the power-law normalization
\begin{align}
\zeta_{ik}^{(x)}
=
\frac{|B_{ik}|^x}
{\sum_{\ell\in\mathcal N_i}|B_{i\ell}|^x},
\qquad
k\in\mathcal N_i,\quad x\geq 0 .
\label{eq:zeta_powerlaw}
\end{align}
For this choice, given $\lambda_i>0$, the node-to-node gains are
\begin{align}
\gamma_{ik}^{(x)}
=
\frac{k_{\mathrm Q_i}V_i^{\rm s}}
{\tau_{\mathrm Q_i}\lambda_i}
\left(
\sum_{\ell\in\mathcal N_i}|B_{i\ell}|^x
\right)
|B_{ik}|^{1-x},
\quad k\in\mathcal N_i .
\label{eq:gamma_powerlaw}
\end{align}
The variable $x$ controls how strongly the normalization weights in \eqref{eq:zeta_powerlaw} favor stronger
susceptive couplings.

The uniform normalization is recovered at $x=0$. Since $|B_{ik}|>0$ for
$k\in\mathcal N_i$, \eqref{eq:zeta_powerlaw} gives
\begin{align}
\zeta_{ik}^{(0)}
=
\frac{1}{|\mathcal N_i|},
\qquad k\in\mathcal N_i,
\label{eq:zeta_uniform}
\end{align}
and hence
\begin{align}
\gamma_{ik}^{(0)}
=
\frac{
|\mathcal N_i|\,k_{\mathrm Q_i}V_i^{\rm s}|B_{ik}|
}{
\tau_{\mathrm Q_i}\lambda_i
},
\qquad k\in\mathcal N_i .
\label{eq:gamma_uniform}
\end{align}
This choice treats all incoming neighbors of node $i$ equally in the
normalization weights. The gain preserves the edge dependence through $|B_{ik}|$, but
it also produces the degree factor $|\mathcal N_i|$. Consequently, high-degree
nodes may be penalized even when many of their incident couplings are weak.

The proportional normalization is recovered at $x=1$. In this case,
\eqref{eq:zeta_powerlaw} gives
\begin{align}
\zeta_{ik}^{(1)}
=
\frac{|B_{ik}|}
{\sum_{\ell\in\mathcal N_i}|B_{i\ell}|},
\qquad k\in\mathcal N_i,
\label{eq:zeta_proportional}
\end{align}
and hence
\begin{align}
\gamma_{ik}^{(1)}
=
\frac{k_{\mathrm Q_i}V_i^{\rm s}}
{\tau_{\mathrm Q_i}\lambda_i}
\sum_{\ell\in\mathcal N_i}|B_{i\ell}|,
\qquad k\in\mathcal N_i .
\label{eq:gamma_proportional}
\end{align}
As the normalization weights are proportional to the incident susceptive couplings, the explicit dependence of the node-to-node gains on the individual
edge strength $|B_{ik}|$ is neutralized. Consequently, all incoming gains to the same receiving node $i$ are equal. Values
$0<x<1$ interpolate between the uniform and proportional normalization, reducing the
degree-scaled worst-neighbor effect while still retaining some dependence on
the individual edge strength $|B_{ik}|$. \hfill $\triangle$
\end{rmk}

\subsection{Cluster-Based Stability Certification}

\begin{thm}[Cluster-based stability criterion]
\label{thm:hierarchical_voltage_stability}
Let $\mathcal V=\{1,\ldots,n\}$ denote the set of grid-forming inverter nodes.
Consider the voltage subsystem
\eqref{eq:linearized_volt_dynamics}-\eqref{eq:Di_def}, obtained under
Assumptions \ref{ass_lossless}-\ref{ass:small_angle_decoupling}. Assume that, for every
$i\in\mathcal V$,
\begin{align}
\tau_{\mathrm{Q}_i},k_{\mathrm{Q}_i}>0,
\qquad
1+k_{\mathrm{Q}_i}D_i>0,
\label{eq:kQ_bound_hierarchical}
\end{align}
where $D_i$ is given by \eqref{eq:Di_def}. Consequently,
\begin{align*}
\lambda_i
=
\frac{1+k_{\mathrm{Q}_i}D_i}{\tau_{\mathrm{Q}_i}}
>0,
\qquad i\in\mathcal V .
\end{align*}
Let the node-to-node gains $\gamma_{ik}\geq 0$, the cluster-to-cluster gains
$\bar\gamma_{\alpha\beta}\geq 0$, and the effective inter-cluster neighbor sets
$\mathcal N_\alpha^{\rm cluster}$ be defined as in
Definition \ref{def:hierarchical_voltage_gains}. Suppose that the following
conditions hold.

\begin{itemize}
\item[(i)] \textbf{Intra-cluster cycle condition.} For each cluster $\mathcal A_\alpha$, every directed cycle contained entirely in $\mathcal A_\alpha$
satisfies
\begin{align}
\gamma_{j_1j_2}\gamma_{j_2j_3}\cdots\gamma_{j_rj_1}<1,
\label{eq:internal_cluster_cycle_condition}
\end{align}
for all $r\in\{2,\ldots,|\mathcal A_\alpha|\}$ and all pairwise distinct
$j_1,\ldots,j_r\in\mathcal A_\alpha$.

\item[(ii)] \textbf{Inter-cluster path condition.} For every cluster $\mathcal A_\alpha$,
\begin{align}
\max_{\beta\in\mathcal N_\alpha^{\rm cluster}}
\bar\gamma_{\alpha\beta}
<1.
\label{eq:inter_cluster_gain_condition}
\end{align}
\end{itemize}
Then the origin $\delta\boldsymbol V=\boldsymbol 0_n$ of \eqref{eq:linearized_volt_dynamics}-\eqref{eq:Di_def}, where
$\delta\boldsymbol V:=\operatorname{col}(\delta V_i)_{i\in\mathcal V}
\in\mathbb R^n$, is exponentially stable. That is, there exist constants
$\lambda_{\rm v}>0$ and $M_{\rm v}>0$ such that, for all $t\geq t_0$,
\begin{align}
\|\delta\boldsymbol V(t)\|_2
\leq
M_{\rm v} e^{-\lambda_{\rm v}(t-t_0)}
\|\delta\boldsymbol V(t_0)\|_2 .
\label{eq:eiss}
\end{align}
\end{thm}

The proof is provided in the Appendix.

\begin{rmk} \rm
    Although Theorem \ref{thm:hierarchical_voltage_stability} is stated for the droop-controlled GFM voltage subsystem \eqref{eq:linearized_volt_dynamics}-\eqref{eq:Di_def}, the same cluster-based certificate applies to any networked scalar linear system of the form
\begin{align}\label{eq:general_dynamics}
        \dot{x}_i = -\lambda_i x_i + \sum_{k\in\mathcal N_i}a_{ik}x_k, \quad i\in\{1,\ldots,n\},
    \end{align}
   where $\lambda_i>0$ and $a_{ik}\neq0$ for
$k\in\mathcal N_i$. For any positive normalization weights $\zeta_{i\ell}$ satisfying $\sum_{\ell\in\mathcal N_i}\zeta_{i\ell}=1$, the corresponding node-to-node gains are defined by
    \begin{align*}
        \gamma_{ik} = \frac{|a_{ik}|}{\lambda_i \zeta_{ik}}.
    \end{align*}
The same intra-cluster cycle and inter-cluster path conditions then provide a sufficient condition for exponential stability. Hence, the framework is not restricted to droop-controlled GFM inverters and can also accommodate other devices, including suitable dynamic loads, whenever their linearized nodal dynamics admit the form \eqref{eq:general_dynamics}.
\end{rmk}

The conditions of Theorem \ref{thm:hierarchical_voltage_stability} can be
verified locally, with each cluster $\mathcal A_\alpha$
assigned to a dedicated local computer. Once the node-to-node gains have been computed locally using the receiving node's controller parameters and one-hop incident electrical data, condition (i) is intra-cluster: enumerating the simple directed cycles
in $\mathcal A_\alpha$ and evaluating the products
$\gamma_{j_1j_2}\cdots\gamma_{j_rj_1}$ require only the node-to-node gains
$\{\gamma_{pq}\}_{p,q\in\mathcal A_\alpha}$. Condition (ii) requires the
cluster-to-cluster gains $\bar\gamma_{\alpha\beta}$ for
$\beta\in\mathcal N_\alpha^{\rm cluster}$; by
Definition \ref{def:hierarchical_voltage_gains}, each such
$\bar\gamma_{\alpha\beta}$ is the maximum product along admissible simple paths
whose intermediate nodes lie in $\mathcal A_\alpha$ and the terminal node
lies in a neighboring cluster $\mathcal A_\beta$. Its evaluation thus
requires only the intra-cluster gains $\{\gamma_{pq}\}_{p,q\in\mathcal A_\alpha}$
and the boundary gains $\{\gamma_{ik}\}_{i\in\mathcal A_\alpha,\,k\in\mathcal A_\beta}$ connecting $\mathcal A_\beta$ to $\mathcal A_\alpha$, for
$\beta\in\mathcal N_\alpha^{\rm cluster}$. No information about non-neighboring clusters or the global network topology is required. See Fig. \ref{fig:cluster} (top right).

The conditions (i) and (ii) of
Theorem \ref{thm:hierarchical_voltage_stability} lead to the following two intra- and inter-cluster stability indices.

\begin{defn}[Cluster-level stability indices]
\label{def:cluster_stability_indices}
For each cluster $\mathcal A_\alpha$, define the \underline{intra-cluster stability
index}
\begin{align}
\begin{split}
\mathcal C_\alpha
:=
\max\Bigl\{
\prod_{\nu=1}^{r}&\gamma_{j_\nu j_{\nu+1}}
\,:\,
r\in\{2,\ldots,|\mathcal A_\alpha|\},\\
&\quad
j_1,\ldots,j_r\in\mathcal A_\alpha
\text{ pairwise distinct}
\Bigr\},
\label{eq:intra_cluster_index}\end{split}
\end{align}
with the cyclic convention $j_{r+1}:=j_1$.

The \underline{inter-cluster stability index} is defined by 
\begin{align}
\Omega_\alpha
:=
\max_{\beta\in\mathcal N_\alpha^{\rm cluster}}
\bar\gamma_{\alpha\beta}.
\label{eq:inter_cluster_index}
\end{align}
 The node-to-node gains $\gamma_{ik}\geq 0$ and cluster-to-cluster gains $\bar\gamma_{\alpha\beta}\geq 0$ are as defined in Definition \ref{def:hierarchical_voltage_gains}. See Fig. \ref{fig:cluster} (top left).
\end{defn}

In terms of these indices, conditions (i) and (ii) of
Theorem \ref{thm:hierarchical_voltage_stability} respectively read
\begin{align}
\mathcal C_\alpha<1
\qquad\text{and}\qquad
\Omega_\alpha<1,
\qquad
\forall \alpha\in\{1,\ldots,m\}.
\label{eq:index_conditions}
\end{align} 
The index $\mathcal C_\alpha$ captures the worst-case cyclic
gain along simple directed cycles in $\mathcal A_\alpha$, while $\Omega_\alpha$
captures the worst-case gain along simple paths
that start in $\mathcal A_\alpha$ and terminate in a cluster
$\mathcal A_\beta$ with $\beta\in\mathcal N_\alpha^{\rm cluster}$.

\begin{rmk}\rm 
Consider the trivial partition consisting of a single cluster corresponding to the entire network. In this case, there are no neighboring clusters, and hence the inter-cluster path condition in Theorem \ref{thm:hierarchical_voltage_stability} is vacuously satisfied. The intra-cluster cycle condition, however, applies to every simple directed cycle in the network. Therefore, Theorem \ref{thm:hierarchical_voltage_stability} reduces to a centralized cyclic small-gain sufficient condition \eqref{eq:internal_cluster_cycle_condition}  for exponential stability of the voltage subsystem. See Proposition \ref{thm:centralized_small_gain_certificate} in the Appendix. Moreover, because $\boldsymbol A_{\rm v}$ is Metzler, its Hurwitz stability implies $\lambda_i>0$ for all $i\in\mathcal V$. A Perron-vector-based normalization can then be chosen so that the  single-cluster certificate, consisting of the conditions $\lambda_i>0$ and the directed cycle conditions \eqref{eq:internal_cluster_cycle_condition}, is \textit{equivalent} to the Hurwitz stability of $\boldsymbol A_{\rm v}$. This normalization requires global network information, and therefore is specific to the centralized setting.
\end{rmk}

\subsection{Decentralized Stability Certification}
\label{sec:volt_stab_decentralized}

We now extract a decentralized node-level stability condition from
Theorem \ref{thm:hierarchical_voltage_stability} by specializing the cluster
partition to the singleton case, where each cluster consists of a single
node. Specifically, take
\begin{align*}
m=n,
\qquad
\mathcal A_i=\{i\},
\qquad
i\in\mathcal V.
\end{align*}
Here, $m$ is the number of pairwise disjoint nonempty clusters; see Definition \ref{def:hierarchical_voltage_gains}. Under this partition, condition (i) of
Theorem \ref{thm:hierarchical_voltage_stability} is vacuous, since every
cluster $\mathcal A_i=\{i\}$ is a singleton and therefore contains no directed
cycle with $r\geq 2$ pairwise distinct nodes. To analyze condition (ii),
observe that for any distinct $\alpha,\beta\in\{1,\ldots,n\}$ with
$\mathcal A_\alpha=\{i\}$ and $\mathcal A_\beta=\{k\}$, the admissible paths in
Definition \ref{def:hierarchical_voltage_gains} reduce to a single one-step
path from $i$ to $k$. The cluster-to-cluster gain therefore collapses to
$\bar\gamma_{\alpha\beta}=\gamma_{ik}$, and the effective inter-cluster neighbor
set $\mathcal N_\alpha^{\rm cluster}$ coincides with the actual neighborhood
$\mathcal N_i$. Condition (ii) of
Theorem \ref{thm:hierarchical_voltage_stability} thus reduces to a per-node
small-gain condition involving only $\gamma_{ik}$ for $k\in\mathcal N_i$. We
formally express this consequence in the following corollary.

\begin{cor}[Decentralized stability criterion]
\label{cor:decentralized_voltage_stability} Let $\mathcal V=\{1,\ldots,n\}$ denote the set of grid-forming inverter nodes. Consider the voltage subsystem
\eqref{eq:linearized_volt_dynamics}-\eqref{eq:Di_def}, obtained under
Assumptions \ref{ass_lossless}-\ref{ass:small_angle_decoupling}. Assume that, for every $i\in\mathcal V$,
$\tau_{\mathrm Q_i}$ and $k_{\mathrm Q_i}$ satisfy
\eqref{eq:kQ_bound_hierarchical}. Let the node-to-node gains
$\gamma_{ik}\ge 0$ be defined as in
Definition \ref{def:hierarchical_voltage_gains}. If
\begin{align}
\max_{k\in\mathcal N_i}\gamma_{ik}<1,
\qquad
\forall i\in\mathcal V,
\label{eq:stability_index}
\end{align}
then the voltage subsystem
\eqref{eq:linearized_volt_dynamics}-\eqref{eq:Di_def} is exponentially stable, and the estimate \eqref{eq:eiss} holds for all $t\ge t_0$.
\end{cor}

\begin{rmk}\label{rem:kQi_bound}\rm
A more explicit sufficient condition ensuring that both \eqref{eq:kQ_bound_hierarchical} and \eqref{eq:stability_index} are satisfied can be derived via conservative bounding arguments. In particular, it can be shown that the conditions of Corollary \ref{cor:decentralized_voltage_stability} hold if
\begin{align*}
\begin{split}
0 < k_{\mathrm Q_i} < \Big(|D_i|
+ V_i^{\rm s} \max_{k \in \mathcal N_i} \big(|B_{ik}|\zeta_{ik}^{-1}\big) 
\Big)^{-1},
\end{split}
\end{align*}
for every $i\in\mathcal V$, where $D_i$ is given by \eqref{eq:Di_def}. \hfill $\triangle$ 
\end{rmk}

The condition \eqref{eq:stability_index} of
Corollary \ref{cor:decentralized_voltage_stability} can be verified by each
inverter node $i\in\mathcal V$ independently. Indeed, the quantity $\max_{k\in\mathcal N_i}\gamma_{ik}$
depends only
on the node-to-node gains $\{\gamma_{ik}\}_{k\in\mathcal N_i}$. Once a normalization choice is adopted,
these gains are computed from the local controller and operating-point data
$k_{\mathrm Q_i},V_i^{\rm s}$ at node $i$, the
self-susceptance $|B_{ii}|$, and the one-hop incident electrical data
$\{|B_{i\ell}|,V_\ell^{\rm s}\}_{\ell\in\mathcal N_i}$. Consequently,
each inverter can certify its own contribution to network voltage stability
using only its own information and one-hop information from its
immediate neighbors. See Fig. \ref{fig:nodeFig} (right). The explicit design rule in Remark \ref{rem:kQi_bound} further refines this to a closed-form bound on $k_{\mathrm{Q}_i}$ in terms of locally available
quantities, enabling each inverter to independently select its droop gain in a
manner that, by Corollary \ref{cor:decentralized_voltage_stability}, guarantees
exponential stability of the network voltage dynamics around the given operating point. 

The decentralized condition \eqref{eq:stability_index} of
Corollary \ref{cor:decentralized_voltage_stability} leads to the following node-level stability index.

\begin{defn}[Node-level stability index]
\label{def:node_stability_index}
For each inverter node $i\in\mathcal V$, define the \emph{node-level stability
index}
\begin{align}
\Xi_i
:=
\max_{k\in\mathcal N_i}\gamma_{ik}.
\label{eq:node_stability_index}
\end{align}
The node-to-node gains $\gamma_{ik}\geq 0$ are defined in Definition \ref{def:hierarchical_voltage_gains}. See Fig. \ref{fig:nodeFig} (left). 
\end{defn}

In terms of this index, the condition \eqref{eq:stability_index} of
Corollary \ref{cor:decentralized_voltage_stability} reads
\begin{align*}
\Xi_i<1,
\qquad
\forall i\in\mathcal V .
\end{align*}
The index $\Xi_i$ captures the worst-case incoming gain at node $i$ from its
neighbors in $\mathcal N_i$.

\subsection{Interpretation of Stability Indices}\label{sec:sparsification}

\subsubsection{\textbf{Decentralized node-level stability index}} Using \eqref{eq:gammaik_def}, each incoming gain at node $i$ can be rearranged as
\begin{align*}
\gamma_{ik}
=
\underbrace{
\left(
\frac{k_{\mathrm{Q}_i}V_i^{\rm s}|B_{ik}|}
{\tau_{\mathrm{Q}_i}\zeta_{ik}}
\right)
}_{\text{normalized incoming coupling burden}}
\bigg/
\underbrace{\lambda_i}_{\text{local damping}},
\qquad k\in\mathcal N_i .
\end{align*}
Thus, $\gamma_{ik}$ compares the normalized voltage-coupling burden from node
$k$ to node $i$ with the local damping available at node $i$. Consequently, the node-level stability index $\Xi_i$ defined in \eqref{eq:node_stability_index} represents the largest incoming coupling-to-damping ratio associated with node $i$ under the selected normalization weights. This interpretation can be understood by identifying the role of each term in $\gamma_{ik}$ as follows:
\begin{itemize}
    \item $\lambda_i$ as given by \eqref{eq:lambda_i},\eqref{eq:Di_def} is the local exponential decay rate of the isolated voltage
    dynamics at node $i$; see \eqref{eq:linearized_volt_dynamics}-\eqref{eq:Di_def}. A larger $\lambda_i$ means stronger local damping and
    therefore a smaller incoming gain $\gamma_{ik}$.

\item The factor $k_{\mathrm{Q}_i}V_i^{\rm s}/\tau_{\mathrm{Q}_i}$ scales the local voltage-control sensitivity at node $i$; see \eqref{eq:linearized_volt_dynamics}. The larger this factor, the greater the sensitivity to neighboring voltage deviations. However, $\tau_{\mathrm{Q}_i}$ cancels from the gain $\gamma_{ik}$ via $\lambda_i$.

    \item $|B_{ik}|$ measures the strength of the susceptive coupling
    between nodes $i$ and $k$. Larger $|B_{ik}|$ means that voltage deviations at
    node $k$ have a stronger effect on the voltage dynamics of node $i$.

    \item $\zeta_{ik}$ is a normalization weight assigned to the incoming voltage-coupling channel from node $k$ to node $i$. The constraint $\sum_{k\in\mathcal N_i}\zeta_{ik}=1$ makes these normalization weights relative across all incoming channels at node $i$.  See Remark \ref{rmk:zeta_choices} for details.
\end{itemize}

Accordingly, the condition $\Xi_i<1$ means that, under the chosen normalization weights, the local damping at node $i$ is large enough, relative to its voltage-control sensitivity and incoming network couplings, to dominate the strongest certified incoming channel. If $\Xi_i\geq 1$, the decentralized certificate fails at node $i$ for the selected weights; this does not necessarily imply instability, and it may still be possible to certify stability using a different choice of normalization weights or the cluster-based stability criterion.

The effect of the normalization choice is especially transparent for the decentralized certificate. Under proportional normalization
\eqref{eq:zeta_proportional},\eqref{eq:gamma_proportional}, and \eqref{eq:node_stability_index} give
\begin{align}
\Xi_i^{\rm prop}
=
\frac{k_{\mathrm{Q}_i}V_i^{\rm s}}
{\tau_{\mathrm{Q}_i}\lambda_i}
\sum_{\ell\in\mathcal N_i}|B_{i\ell}|.
\label{eq:prop_index}
\end{align}
By contrast, under uniform normalization \eqref{eq:zeta_uniform},
\begin{align}
\Xi_i^{\rm unif}
=
\frac{k_{\mathrm{Q}_i}V_i^{\rm s}}
{\tau_{\mathrm{Q}_i}\lambda_i}
|\mathcal N_i|
\max_{k\in\mathcal N_i}|B_{ik}|.
\label{eq:unif_index}
\end{align}
Since $\sum_{\ell\in\mathcal N_i}|B_{i\ell}|
\leq
|\mathcal N_i|
\max_{k\in\mathcal N_i}|B_{ik}|$, the proportional normalization is never more conservative than the uniform normalization for the decentralized certificate. The power-law family \eqref{eq:zeta_powerlaw} provides intermediate choices.

In fact, the proportional choice \eqref{eq:zeta_proportional} is optimal among
all admissible normalization weights satisfying \eqref{eq:zeta_allocation}, in
the sense of minimizing the worst incoming gain
$\max_{k\in\mathcal N_i}\gamma_{ik}$ at each node $i$. To see this, let
$M_{ik}:=k_{\mathrm{Q}_i}V_i^{\rm s}|B_{ik}|/
(\tau_{\mathrm{Q}_i}\lambda_i)$. For any admissible weights satisfying \eqref{eq:zeta_allocation},
\begin{align*}
\max_{k\in\mathcal N_i}\frac{M_{ik}}{\zeta_{ik}}
\geq \sum_{k\in\mathcal N_i}M_{ik},
\end{align*}
where equality is attained by
$\zeta_{ik}=M_{ik}/\sum_{\ell\in\mathcal N_i}M_{i\ell}$, which is precisely
the proportional normalization \eqref{eq:zeta_proportional}. Hence,
proportional normalization minimizes $\Xi_i$ at each node.

The proportional certificate also admits a familiar matrix interpretation. Provided $\lambda_i>0$, the condition
$\Xi_i^{\rm prop}<1$ is equivalent to
\begin{align*}
\lambda_i
>
\sum_{k\in\mathcal N_i}
\frac{k_{\mathrm Q_i}V_i^{\rm s}|B_{ik}|}{\tau_{\mathrm Q_i}},
\qquad i\in\mathcal V,
\end{align*}
which is the strict row diagonal dominance of $\boldsymbol A_{\rm v}$ with negative diagonal entries. Hence, by the Gershgorin circle theorem, all eigenvalues of $\boldsymbol A_{\rm v}$ lie in
the open left half-plane, and $\boldsymbol A_{\rm v}$ is Hurwitz
\cite{horn2012matrix}.

\begin{rmk}[Connection with the decentralized criterion of \cite{siahaan2024decentralized}]\rm
Under the no-shunt condition
$|B_{ii}|=\sum_{k\in\mathcal N_i}|B_{ik}|$, the decentralized criterion
proposed in \cite{siahaan2024decentralized} is
\begin{align}
k_{\mathrm Q_i}|B_{ii}|
\big(
\max_{k\in\mathcal N_i}V_k^{\rm s}-V_i^{\rm s}
\big)
<1,
\qquad i\in\mathcal V.
\label{eq:prop_max_voltage_bound}
\end{align}
Within the present framework, this condition is a sufficient condition for the
decentralized certificate under proportional normalization. Indeed,
\eqref{eq:prop_max_voltage_bound} guarantees both $\lambda_i>0$ and
$\Xi_i^{\rm prop}<1$. It is generally not equivalent to
$\Xi_i^{\rm prop}<1$, since it replaces the weighted aggregate of neighboring
voltage differences by an upper bound based on the largest neighboring voltage
difference.
\hfill\(\triangle\)
\end{rmk}

\subsubsection{\textbf{Cluster-level stability indices}} The cluster-based stability indices \eqref{eq:intra_cluster_index} and \eqref{eq:inter_cluster_index} separate the network's voltage interactions into intra-cluster effects, captured by
$\mathcal C_\alpha$, and inter-cluster effects, captured by $\Omega_\alpha$.

The intra-cluster index $\mathcal C_\alpha$ is the largest product of node-to-node
gains around any simple directed cycle contained entirely within cluster
$\mathcal A_\alpha$. Each factor $\gamma_{j_\nu j_{\nu+1}}$ along such a cycle
is itself the incoming coupling-to-damping ratio at node $j_\nu$ from node
$j_{\nu+1}$, in the sense established above for the node-level index. The
condition $\mathcal C_\alpha<1$ therefore requires that, for every feedback loop closing inside
$\mathcal A_\alpha$, the cumulative product of incoming coupling-to-damping
ratios along the loop is strictly less than one. Equivalently, no internal
feedback path in $\mathcal A_\alpha$ admits net amplification: a voltage deviation traversing such a loop returns to its starting node attenuated.

The inter-cluster index $\Omega_\alpha$ admits a similar physical reading. By
\eqref{eq:cluster_gain_def}, the cluster-to-cluster gain $\bar\gamma_{\alpha\beta}$ is
the largest product of node-to-node gains along any simple path that
originates at a node in $\mathcal A_\alpha$, propagates through distinct nodes
of $\mathcal A_\alpha$, and terminates at a single node in
$\mathcal A_\beta$. Since $\gamma_{ik}$ encodes the gain from node $k$ to node
$i$, this path quantifies the strongest channel through which a voltage
deviation originating in the source cluster $\mathcal A_\beta$ can influence the
receiving cluster $\mathcal A_\alpha$, accounting for amplification along
admissible propagation routes within $\mathcal A_\alpha$. The index
$\Omega_\alpha$ extracts the worst case among these incoming influences across all clusters in $\mathcal N_\alpha^{\rm cluster}$. The condition $\Omega_\alpha<1$ thus requires that the strongest inter-cluster influence feeding
into $\mathcal A_\alpha$ has net gain strictly less than one.

\subsubsection{\textbf{Comparison}}\label{sec:comparison_partition_normalization} Relative to the decentralized condition $\Xi_i<1$, the joint cluster-based conditions $\mathcal C_\alpha<1$ and $\Omega_\alpha<1$ replace per-node dominance by per-cluster cycle and path conditions. For any fixed admissible choice of normalization weights, the decentralized condition $\Xi_i<1$ for every $i\in\mathcal V$ implies both $\mathcal C_\alpha<1$ and $\Omega_\alpha<1$ for every $\alpha$. Indeed, every factor $\gamma_{i_\nu i_{\nu+1}}$ appearing in the cycle product defining $\mathcal C_\alpha$ or the path product defining $\bar\gamma_{\alpha\beta}$ is
bounded above by $\Xi_{i_\nu}<1$.  Hence all such products are strictly less than one.

The reverse implication does not hold. A node $i$ may satisfy $\Xi_i\geq 1$, for instance due to a dominant
neighbor $k\in\mathcal N_i$ with $\gamma_{ik}\geq 1$, while the cluster $\mathcal A_\alpha\ni i$ still satisfies both $\mathcal C_\alpha<1$ and $\Omega_\alpha<1$. However, for any fixed normalization, this is possible only if every channel responsible for the decentralized failure is internalized. Indeed, let
$i\in\mathcal A_\alpha$. If a node $k\in\mathcal N_i$ satisfying $\gamma_{ik}\geq1$ belongs to another cluster $\mathcal A_\beta$, the direct channel connecting $i$ and $k$ is an admissible inter-cluster path contributing to
$\bar\gamma_{\alpha\beta}$. Thus, $\Omega_\alpha\geq\bar\gamma_{\alpha\beta}\geq\gamma_{ik}\geq1$, and the
cluster certificate also fails. Clustering can therefore reduce conservatism only by placing all neighbors associated with the violating gains in the same cluster as the limiting node. The corresponding channels are then no longer assessed as direct one-edge inter-cluster paths, but instead enter intra-cluster cycle products and, where applicable, longer inter-cluster path products. These products may remain below one even when $\Xi_i\geq1$. Whether this reduction can be realized depends on how the selected normalization distributes the limiting gains among the neighbors of the failing node.

Under uniform normalization, the incoming gains at a particular node depend on the
coupling strengths $|B_{ik}|$; see \eqref{eq:gamma_uniform}. They are therefore generally unequal. Thus, if the uniform decentralized certificate fails at node $i$, that is, $\Xi_i^{\rm unif}\geq 1$,  only a subset of the incoming gains may satisfy $\gamma_{ik}\geq 1$. Consequently, clustering  may improve upon the decentralized certificate by placing all neighbors associated with these violating gains in the same cluster as the limiting node, without necessarily internalizing every neighbor of that node.

Under proportional normalization, the situation is different. From \eqref{eq:gamma_proportional} and \eqref{eq:prop_index}, all incoming gains to
a fixed receiving node $i$ are equal, that is,
$\gamma_{ik}=\Xi_i^{\rm prop}$ for every $k\in\mathcal N_i$. Hence, if the proportional decentralized certificate fails at node $i$, then $\gamma_{ik}\geq1$ for every $k\in\mathcal N_i$. Consequently, clustering can improve upon the decentralized certificate only by placing the limiting node and \textit{all} of its neighbors in the same cluster.

\section{Stability of the Angle-Frequency Subsystem}\label{sec:angle}

In this section, we analyze the stability of the angle-frequency subsystem \eqref{eq:af_theta_dyn}-\eqref{eq:af_omega_dyn_n}. To this end, we impose the following additional assumption.

\begin{as}[Reciprocal susceptive couplings]
\label{as:reciprocal_angle_frequency}
For every edge $\{i,k\}\in\mathcal E$, the  susceptive coupling is reciprocal:
\begin{align*}
B_{ik}=B_{ki}.
\end{align*}
\end{as}
For all $i,k\in\mathcal V$ with $i\neq k$, define the  weight
\begin{align*}
H_{ik}
:=
\begin{cases}
|B_{ik}|V_i^{\rm s}V_k^{\rm s}, & \{i,k\}\in\mathcal E,\\
0, & \{i,k\}\notin\mathcal E.
\end{cases}
\end{align*}
Under Assumption \ref{ass_lossless}, for every edge $\{i,k\}\in\mathcal E$,
one has $G_{ik}=0$. Since $Y_{ik}\neq 0$ for any $\{i,k\}\in\mathcal E$, we have $B_{ik}\neq 0$. Furthermore, it follows from Assumption \ref{ass_sync_operating_point} that
$V_i^{\rm s},V_k^{\rm s}>0$. Hence,
\begin{align*}
H_{ik}>0,
\qquad
\forall \{i,k\}\in\mathcal E,
\end{align*}
and $H_{ik}=0$ whenever $\{i,k\}\notin\mathcal E$. Moreover, under
Assumption \ref{as:reciprocal_angle_frequency},
\begin{align*}
H_{ik}=H_{ki},
\qquad
\forall i,k\in\mathcal V,\ i\neq k.
\end{align*}

Define the symmetric weighted Laplacian $\boldsymbol L_H\in\mathbb R^{n\times n}$ by
\begin{align*}
(\boldsymbol L_H)_{ii}
&:=
\sum_{k\in\mathcal V\setminus\{i\}}H_{ik},
\\
(\boldsymbol L_H)_{ik}
&:=
-H_{ik},
\qquad i\neq k.
\end{align*}
Partition $\boldsymbol L_H$ according to the non-reference nodes
$\mathcal V_0$ and the reference node $n$ as
\begin{align}
\boldsymbol L_H
=
\begin{bmatrix}
\boldsymbol L_{00} & \boldsymbol l_{0n}\\
\boldsymbol l_{n0} & l_{nn}
\end{bmatrix},
\label{eq:LH_partition}
\end{align}
where $\boldsymbol L_{00}\in\mathbb R^{(n-1)\times(n-1)}$,
$\boldsymbol l_{0n}\in\mathbb R^{n-1}$, and
$\boldsymbol l_{n0}\in\mathbb R^{1\times(n-1)}$.

Define $\delta\boldsymbol\theta :=
\operatorname{col}(\delta\theta_i)_{i\in\mathcal V_0}
\in\mathbb R^{n-1},\, 
\delta\boldsymbol\omega_0
:=
\operatorname{col}(\delta\omega_i)_{i\in\mathcal V_0}
\\\in\mathbb R^{n-1},\, 
\delta\boldsymbol\omega :=
\operatorname{col}(\delta\boldsymbol\omega_0,\delta\omega_n)
\in\mathbb R^n.$ Also define $\boldsymbol T_0 :=
\operatorname{diag}(\tau_{\mathrm{P}_i})_{i\in\mathcal V_0},\, 
\boldsymbol K_0 :=
\operatorname{diag}(k_{\mathrm{P}_i})_{i\in\mathcal V_0}.$ Then the decoupled small-angle reduced angle-frequency dynamics \eqref{eq:af_theta_dyn}-\eqref{eq:af_omega_dyn_n} can be written as
\begin{align}
\delta\dot{\boldsymbol\theta}
&=
\delta\boldsymbol\omega_0
-
\boldsymbol 1_{n-1}\delta\omega_n,
\label{eq:af_reduced_theta_vector}
\\
\boldsymbol T_0\delta\dot{\boldsymbol\omega}_0
&=
-\delta\boldsymbol\omega_0
-
\boldsymbol K_0\boldsymbol L_{00}\delta\boldsymbol\theta,
\label{eq:af_reduced_omega0_vector}
\\
\tau_{\mathrm{P}_n}\delta\dot\omega_n
&=
-\delta\omega_n
-
k_{\mathrm{P}_n}\boldsymbol l_{n0}\delta\boldsymbol\theta.
\label{eq:af_reduced_omegan_vector}
\end{align}

The symmetric weighted-Laplacian structure of the angle-frequency subsystem enables an energy-based stability analysis, with an energy function that combines the reduced angle potential and frequency kinetic terms. 

\begin{thm}[Angle-frequency stability]
\label{thm:af_reduced_energy_stability}
Consider the angle-frequency subsystem
\eqref{eq:af_reduced_theta_vector}-\eqref{eq:af_reduced_omegan_vector},
obtained under Assumptions \ref{ass_lossless}-\ref{as:reciprocal_angle_frequency}. If $k_{\mathrm{P}_i},\tau_{\mathrm{P}_i}>0,
i\in\mathcal V,$ then the origin of \eqref{eq:af_reduced_theta_vector}-\eqref{eq:af_reduced_omegan_vector}
\begin{align*}
(\delta\boldsymbol\theta,\delta\boldsymbol\omega)
=
(\boldsymbol 0_{n-1},\boldsymbol 0_n),
\end{align*}
is exponentially stable. Consequently, there exist constants $M_{\rm af}>0$ and
$\lambda_{\rm af}>0$ such that, for all $t\ge t_0$,
\begin{align}
\left\|
\begin{bmatrix}
\delta\boldsymbol\theta(t)\\
\delta\boldsymbol\omega(t)
\end{bmatrix}
\right\|_2
\le
M_{\rm af}e^{-\lambda_{\rm af}(t-t_0)}
\left\|
\begin{bmatrix}
\delta\boldsymbol\theta(t_0)\\
\delta\boldsymbol\omega(t_0)
\end{bmatrix}
\right\|_2.
\label{eq:af_reduced_exp_estimate}
\end{align}
\end{thm}

\noindent\textbf{Proof.} By Assumption \ref{ass_sync_operating_point},
$\boldsymbol V^{\rm s}\in\mathbb R_{>0}^n$. Moreover, since we work under
Assumption \ref{ass_lossless}, for every edge $\{i,k\}\in\mathcal E$ we have
$G_{ik}=0$ and $Y_{ik}\neq 0$, hence $B_{ik}\neq 0$. Therefore, for every
$\{i,k\}\in\mathcal E$,
$H_{ik}
=
|B_{ik}|V_i^{\rm s}V_k^{\rm s}
>0.$ Thus the weighted graph associated with $\boldsymbol L_H$ has the same edge set
as the susceptance graph $\mathcal G$. Since $\mathcal G$ is connected by
Assumption \ref{as:connected_reduced_graph}, the weighted graph associated with
$\boldsymbol L_H$ is connected. Furthermore, it follows from Assumption \ref{as:reciprocal_angle_frequency} that $
\boldsymbol L_H=\boldsymbol L_H^\top$. Consequently, $\boldsymbol L_H$ is the symmetric weighted Laplacian of a
connected undirected graph with positive edge weights. Therefore, $\boldsymbol L_H\succeq 0$ and $\ker(\boldsymbol L_H)=\operatorname{span}\{\boldsymbol 1_n\}$. Moreover, every principal submatrix obtained by deleting one row and the
corresponding column of a connected weighted Laplacian is positive definite. Hence
\begin{align}
\boldsymbol L_{00}\succ 0.
\label{eq:L00_positive_definite}
\end{align}
Because $\boldsymbol L_H\boldsymbol 1_n=\boldsymbol 0_n$, the block partition
\eqref{eq:LH_partition} gives
\begin{align}
\boldsymbol L_{00}\boldsymbol 1_{n-1}
+
\boldsymbol l_{0n}
&=
\boldsymbol 0_{n-1},
\label{eq:L00_l0n_identity}
\\
\boldsymbol l_{n0}\boldsymbol 1_{n-1}
+
l_{nn}
&=
0.
\label{eq:ln0_lnn_identity}
\end{align}
Also, since $\boldsymbol L_H=\boldsymbol L_H^\top$, we have
\begin{align}
\boldsymbol l_{n0}
=
\boldsymbol l_{0n}^\top.
\label{eq:ln0_l0n_transpose}
\end{align}
Consider the Lyapunov function
\begin{align*}
E
:=
\frac{1}{2}
\delta\boldsymbol\theta^\top
\boldsymbol L_{00}
\delta\boldsymbol\theta
+
\frac{1}{2}
\delta\boldsymbol\omega_0^\top
\boldsymbol K_0^{-1}\boldsymbol T_0
\delta\boldsymbol\omega_0
+
\frac{1}{2}
\frac{\tau_{\mathrm{P}_n}}{k_{\mathrm{P}_n}}
(\delta\omega_n)^2.
\end{align*}
By \eqref{eq:L00_positive_definite},
$E$ is positive definite with respect to
$(\delta\boldsymbol\theta,\delta\boldsymbol\omega_0,\delta\omega_n)$. Taking the derivative of $E$ along
\eqref{eq:af_reduced_theta_vector}-\eqref{eq:af_reduced_omegan_vector}, we get
\begin{align*}
\begin{split}
\dot E
&= \delta\boldsymbol\theta^\top
\boldsymbol L_{00}
\left(
\delta\boldsymbol\omega_0
-
\boldsymbol 1_{n-1}\delta\omega_n
\right)
\\
&\quad+
\delta\boldsymbol\omega_0^\top
\boldsymbol K_0^{-1}
\left(
-\delta\boldsymbol\omega_0
-
\boldsymbol K_0\boldsymbol L_{00}\delta\boldsymbol\theta
\right)
\\
&\quad+
\frac{1}{k_{\mathrm{P}_n}}
\delta\omega_n
\left(
-\delta\omega_n
-
k_{\mathrm{P}_n}\boldsymbol l_{n0}\delta\boldsymbol\theta
\right).
\end{split}
\end{align*}
Expanding the terms gives
\begin{align*}
\begin{split}
\dot E
&=
\delta\boldsymbol\theta^\top
\boldsymbol L_{00}
\delta\boldsymbol\omega_0
-
\delta\boldsymbol\theta^\top
\boldsymbol L_{00}
\boldsymbol 1_{n-1}\delta\omega_n
-
\delta\boldsymbol\omega_0^\top
\boldsymbol K_0^{-1}
\delta\boldsymbol\omega_0
\\&\quad-
\delta\boldsymbol\omega_0^\top
\boldsymbol L_{00}\delta\boldsymbol\theta
-
\frac{1}{k_{\mathrm{P}_n}}(\delta\omega_n)^2
-
\delta\omega_n\boldsymbol l_{n0}\delta\boldsymbol\theta.
\end{split}
\end{align*}
Since $\boldsymbol L_{00}=\boldsymbol L_{00}^\top$, the first and fourth terms
cancel. Using \eqref{eq:L00_l0n_identity}, we have
$\boldsymbol L_{00}\boldsymbol 1_{n-1}
=
-\boldsymbol l_{0n}.$ Therefore, $-\delta\boldsymbol\theta^\top
\boldsymbol L_{00}
\boldsymbol 1_{n-1}\delta\omega_n
=
\delta\boldsymbol\theta^\top
\boldsymbol l_{0n}\delta\omega_n =
\delta\omega_n\boldsymbol l_{0n}^\top\delta\boldsymbol\theta =
\delta\omega_n\boldsymbol l_{n0}\delta\boldsymbol\theta,$ where \eqref{eq:ln0_l0n_transpose} was used in the last equality. Hence this
term cancels the term
$-\delta\omega_n\boldsymbol l_{n0}\delta\boldsymbol\theta$. Thus
\begin{align}
\dot E
=
-
\delta\boldsymbol\omega_0^\top
\boldsymbol K_0^{-1}
\delta\boldsymbol\omega_0
-
\frac{1}{k_{\mathrm{P}_n}}(\delta\omega_n)^2
\le 0.
\label{eq:af_reduced_Vdot}
\end{align}
We now characterize the largest invariant set contained in $\{\dot E=0\}$.
From \eqref{eq:af_reduced_Vdot}, $\dot E=0$ implies $\delta\boldsymbol\omega_0=\boldsymbol 0_{n-1},\;
\delta\omega_n=0.$ On an invariant trajectory contained in this set, one must also have $\delta\dot{\boldsymbol\omega}_0=\boldsymbol 0_{n-1}.$ Using \eqref{eq:af_reduced_omega0_vector}, this gives $\boldsymbol 0_{n-1}
=
-\boldsymbol K_0\boldsymbol L_{00}\delta\boldsymbol\theta.$ Since $\boldsymbol K_0$ is nonsingular, $
\boldsymbol L_{00}\delta\boldsymbol\theta
=
\boldsymbol 0_{n-1}.$ Because $\boldsymbol L_{00}\succ0$, it follows that $\delta\boldsymbol\theta
=
\boldsymbol 0_{n-1}.$ Thus the largest invariant set in $\{\dot E=0\}$ is the origin. By LaSalle's invariance principle, the origin of the reduced angle-frequency
subsystem is asymptotically stable. Since
\eqref{eq:af_reduced_theta_vector}-\eqref{eq:af_reduced_omegan_vector} is a
linear finite-dimensional system, asymptotic stability is equivalent to
exponential stability. Therefore, there exist constants $M_{\rm af}>0$ and
$\lambda_{\rm af}>0$ such that \eqref{eq:af_reduced_exp_estimate} holds for all
$t\ge t_0$. This completes the proof.
\hfill $\square$

Unlike the voltage-subsystem conditions in Theorem \ref{thm:hierarchical_voltage_stability} or Corollary \ref{cor:decentralized_voltage_stability}, the angle-frequency stability result requires no additional bounds on the active-power droop gains or filter time constants. Under the stated assumptions, positivity alone,
$k_{\mathrm{P}_i}>0$ and $\tau_{\mathrm{P}_i}>0$ for all $i\in\mathcal V$, is sufficient for exponential stability of the decoupled angle-frequency
subsystem.

\section{Numerical Simulations}\label{sec:simulation}

\subsection{Synthetic Inverter Network}

\begin{figure}
    \centering    \includegraphics[width=0.9\linewidth]{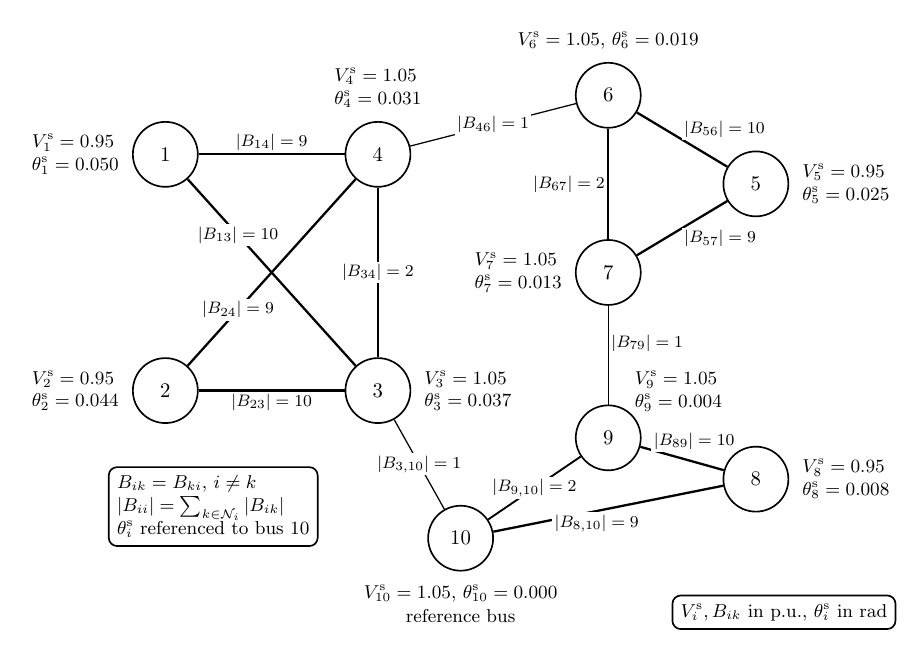}
    \caption{Synthetic 10-bus inverter network.}
    \label{fig:artificial}
\end{figure}

\begin{table*}[t]
\centering
\caption{Critical $k_{\rm Q}$ values for eigenvalue tests and cyclic small-gain certificates for the synthetic 10-bus network considered in Fig. \ref{fig:artificial} under different normalization exponents $x$ in the power-law family.}
\label{tab:kq_critical_values_x_sweep}
\setlength{\tabcolsep}{4pt} 
\renewcommand{\arraystretch}{0.7} 
\begin{tabular}{lccccccccc}
\toprule
 & \multicolumn{9}{c}{Critical value of \(k_{\rm Q}\) in p.u.} \\
\midrule
Volt. subsys. eig.
& \multicolumn{9}{c}{\(6.37\)} \\
Full sys. eig.
& \multicolumn{9}{c}{\(6.27\)} \\
\midrule
norm. expon.
& \(x=0\) & \(x=0.25\) & \(x=0.5\) & \(x=0.75\) & \(x=0.8\) & \(x=0.85\) & \(x=0.9\) & \(x=0.95\) & \(x=1\) \\
\midrule
Decentral.
& \(0.05\) & \(0.11\) & \(0.25\) & \(0.46\) & \(0.48\) & \(0.49\) & \(0.50\) & \(0.51\) & \(0.52\) \\
3-clus.
& \(0.11\) & \(0.17\) & \(0.31\) & \(0.68\) & \(0.84\) & \(1.05\) & \(1.39\) & \(1.96\) & \(3.13\) \\
2-clus.
& \(0.11\) & \(0.17\) & \(0.31\) & \(0.68\) & \(0.84\) & \(1.05\) & \(1.39\) & \(1.96\) & \(3.13\) \\
1-clus. (central.)
& \(0.11\) & \(0.17\) & \(0.31\) & \(0.68\) & \(0.84\) & \(1.05\) & \(1.39\) & \(1.96\) & \(3.13\) \\
\bottomrule
\end{tabular}
\end{table*}

\begin{figure}
    \centering
    \includegraphics[width=0.9\linewidth]{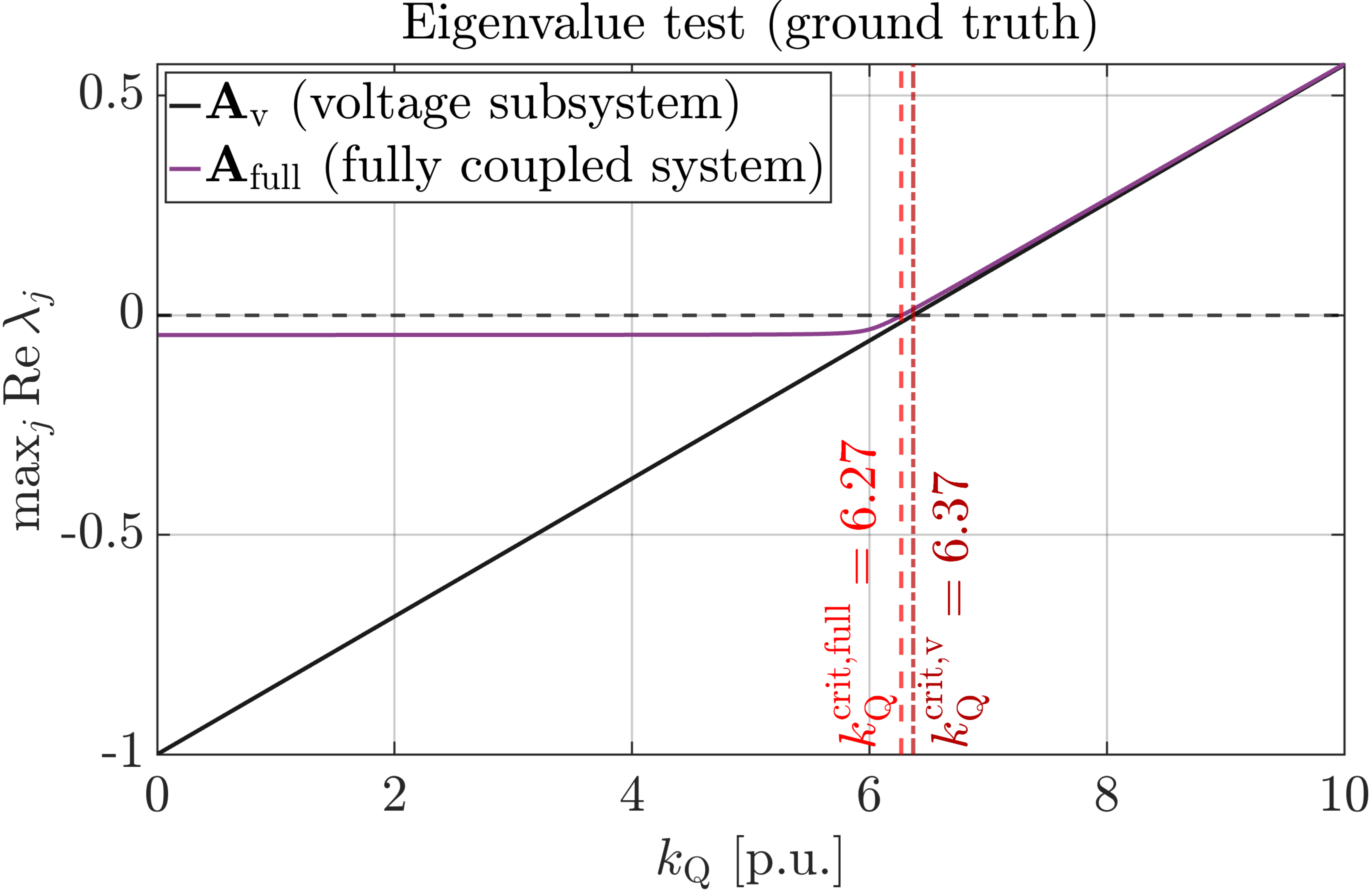}
      \caption{Synthetic 10-bus network: eigenvalue ground truth. Maximum real part of the spectrum vs. the
  homogeneous $k_{\mathrm{Q}}$, for the voltage subsystem
  $\boldsymbol A_{\rm v}$ and for the fully coupled linearization.}
    \label{fig:eigenvalue_artificial}
\end{figure}

\begin{figure}
    \centering
    \includegraphics[width=0.9\linewidth]{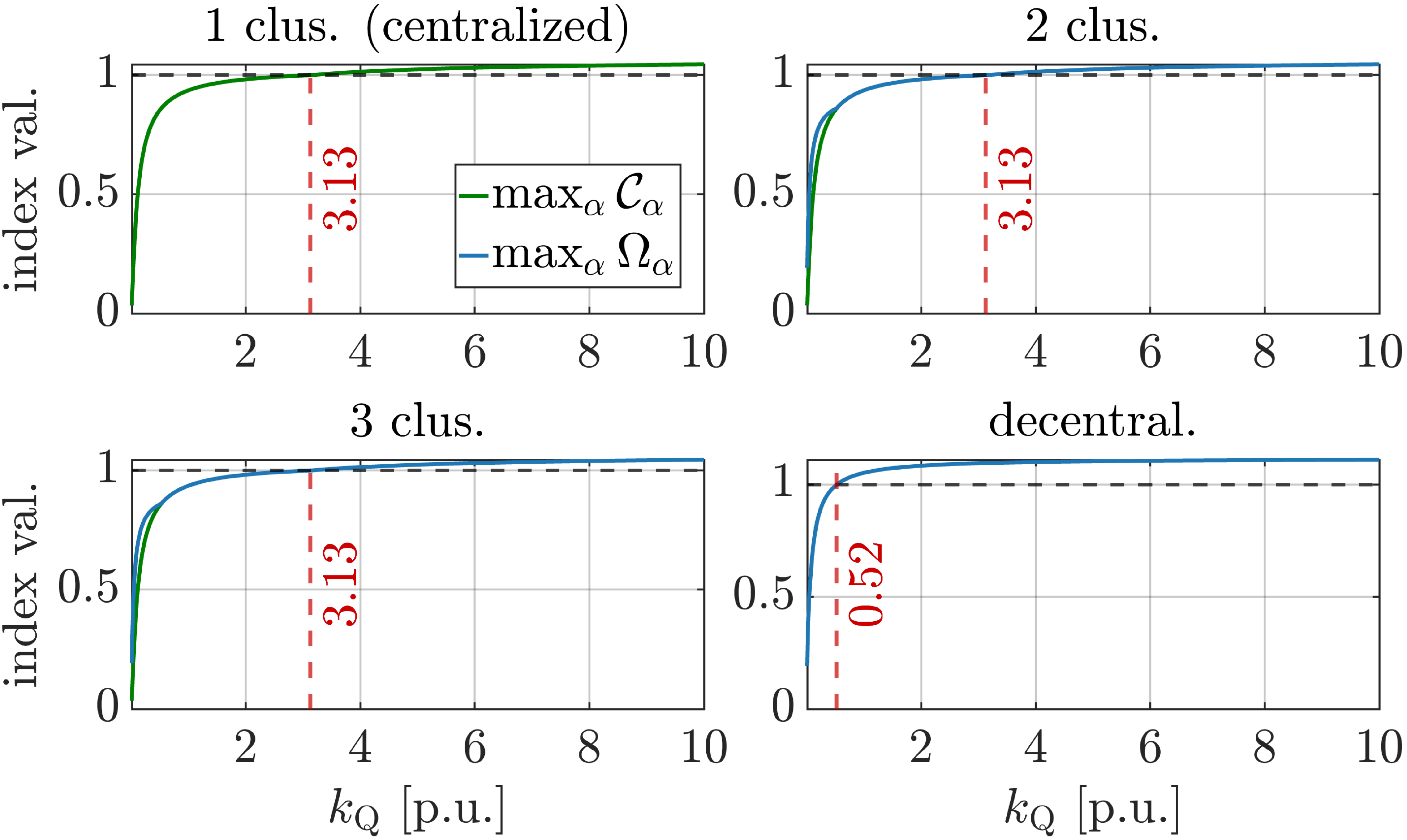}
\caption{Synthetic 10-bus network, proportional normalization ($x=1$): cluster- and node-level stability indices vs. the homogeneous  $k_{\mathrm{Q}}$.
  Each panel corresponds to one partition---the single cluster (centralized), two clusters, three clusters, and the
  singleton partition (decentralized).}
    \label{fig:indices_KQ_prop_artificial}
\end{figure}

\begin{figure}
    \centering
    \includegraphics[width=0.9\linewidth]{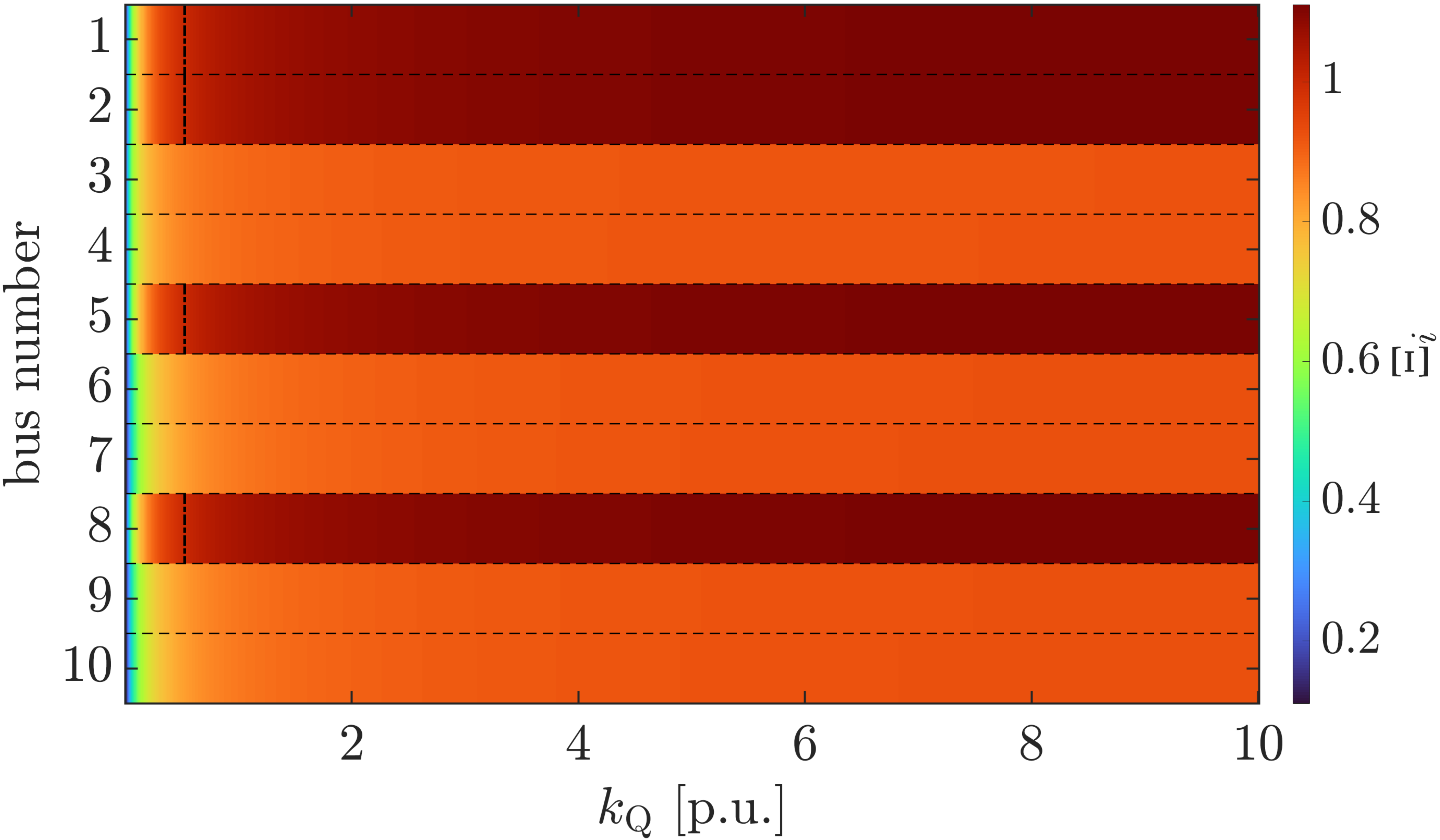}
  \caption{Synthetic 10-bus network, proportional normalization ($x=1$): node-level stability index
  $\Xi_i$
  (Definition \ref{def:node_stability_index}) as a function of $k_{\mathrm{Q}}$.}
    \label{fig:heat_map_Xi_prop_artificial}
\end{figure}

We first evaluate the proposed certificates on the synthetic 10-bus 
network shown in Fig. \ref{fig:artificial}. In addition to the singleton  and single-cluster  partitions, we consider the 2-cluster partition $\big\{\{1,2,3,4\},\{5,6,7,8,9,10\}\big\}$ and the 3-cluster partition $\big\{\{1,2,3,4\},\{5,6,7\},\{8,9,10\}\big\}.$ The filter time constants are set to $\tau_{\mathrm P_i}=\tau_{\mathrm Q_i}=1$s, the active-power/frequency droop gains are set to $k_{\mathrm P_i}=0.05$ p.u., and the homogeneous reactive-power/voltage droop gain $k_{\mathrm Q_i}=k_{\mathrm Q}$ is swept  over $k_{\mathrm Q}\in[0.01,10]$ p.u. with a step size of $0.01$ p.u.

We compare the certificates proposed in Theorem \ref{thm:hierarchical_voltage_stability} and Corollary \ref{cor:decentralized_voltage_stability} against two eigenvalue
ground-truth tests: one based on the
voltage subsystem $\boldsymbol A_{\rm v}$ of
\eqref{eq:linearized_volt_dynamics}-\eqref{eq:Di_def}, and one based on the fully coupled linearization
\eqref{eq:lin_theta_i}-\eqref{eq:lin_V_n}, with the
latter retaining the angle-voltage cross-coupling that
Assumption \ref{ass:small_angle_decoupling} discards. A configuration is declared stable by an eigenvalue test when the corresponding maximum real part of the spectrum is negative. For each certificate, $\lambda_i>0$ is checked at every tested point; otherwise the certificate is declared to fail. Subject to this requirement, the certificate holds when the relevant stability index is below unity. For each certificate, we report results for different exponents $x\in[0,1]$ in the power-law normalization, where $x=0$ gives the uniform
normalization and $x=1$ gives the proportional normalization.

The critical $k_{\mathrm Q}$ values obtained from the tests are summarized in
Table \ref{tab:kq_critical_values_x_sweep}. Over the sweep $k_{\mathrm Q}\in[0.01,10]$ p.u., the voltage subsystem remains Hurwitz
up to $k_{\mathrm Q}=6.37$ p.u., while the fully coupled linearization loses
stability at $k_{\mathrm Q}=6.27$ p.u. Thus, for this example, the decoupled
voltage subsystem closely approximates the full small-signal stability
boundary. This is also evident in Fig. \ref{fig:eigenvalue_artificial}, where the maximum real parts of the voltage-subsystem and full-system eigenvalues cross zero at nearby values of $k_{\mathrm Q}$.

The small-gain certificates are more conservative than the eigenvalue tests, but Table \ref{tab:kq_critical_values_x_sweep} shows that the degree of conservatism depends on both the partition and the normalization. Under uniform normalization, $x=0$, the decentralized certificate remains
valid only up to $k_{\mathrm Q}=0.05$ p.u., whereas the 2-cluster, 3-cluster,
and centralized cyclic certificates remain valid up to
$k_{\mathrm Q}=0.11$ p.u. As $x$ increases, the certified range expands for
all certificates, with a larger improvement for the non-singleton
partitions. Under proportional normalization, $x=1$, the decentralized
certificate fails at $k_{\mathrm Q}=0.52$ p.u., while the 2-cluster, 3-cluster,
and centralized cycle/path certificates all certify stability up to
$k_{\mathrm Q}=3.13$ p.u. Figure \ref{fig:indices_KQ_prop_artificial} shows how these critical values arise from the evolution of the stability indices under the proportional normalization $x=1$. The node-level heat map in
Fig. \ref{fig:heat_map_Xi_prop_artificial} localizes the decentralized
limitation, with buses $1,2,5$, and $8$ being the nodes whose indices $\Xi_i$ reach unity first.

\subsection{New England IEEE 39-Bus Test System}

\begin{table*}[t]
\centering
\caption{Critical \(k_{\rm Q}\) values for eigenvalue tests and cyclic small-gain certificates for the Kron-reduced IEEE 39-bus system under different normalization exponents \(x\) in the power-law family.}
\label{tab:ieee39_kq_critical_values_x_sweep}
\setlength{\tabcolsep}{4pt} 
\renewcommand{\arraystretch}{0.7} 
\begin{tabular}{lcccccccc}
\toprule
 & \multicolumn{8}{c}{Critical value of \(k_{\rm Q}\) in p.u.} \\
\midrule
Volt. subsys. eig.
& \multicolumn{8}{c}{\(0.404\)} \\
Full sys. eig.
& \multicolumn{8}{c}{\(0.335\)} \\
\midrule
norm. expon.
& \(x=0\) & \(x=0.25\) & \(x=0.5\) & \(x=0.6\) & \(x=0.7\) & \(x=0.8\) & \(x=0.9\) & \(x=1\) \\
\midrule
Decentral.
& \(<0.01\) & \(0.016\) & \(0.033\) & \(0.046\) & \(0.066\) & \(0.088\) & \(0.107\) & \(0.132\) \\
3-clus.
& \(<0.01\) & \(0.017\) & \(0.035\) & \(0.049\) & \(0.070\) & \(0.088\) & \(0.107\) & \(0.132\) \\
2-clus.
& \(<0.01\) & \(0.017\) & \(0.035\) & \(0.049\) & \(0.070\) & \(0.088\) & \(0.107\) & \(0.132\) \\
1-clus. (central.)
& \(<0.01\) & \(0.017\) & \(0.035\) & \(0.049\) & \(0.070\) & \(0.109\) & \(0.147\) & \(0.190\) \\
\bottomrule
\end{tabular}
\end{table*}

\begin{figure}
    \centering
    \includegraphics[width=0.9\linewidth]{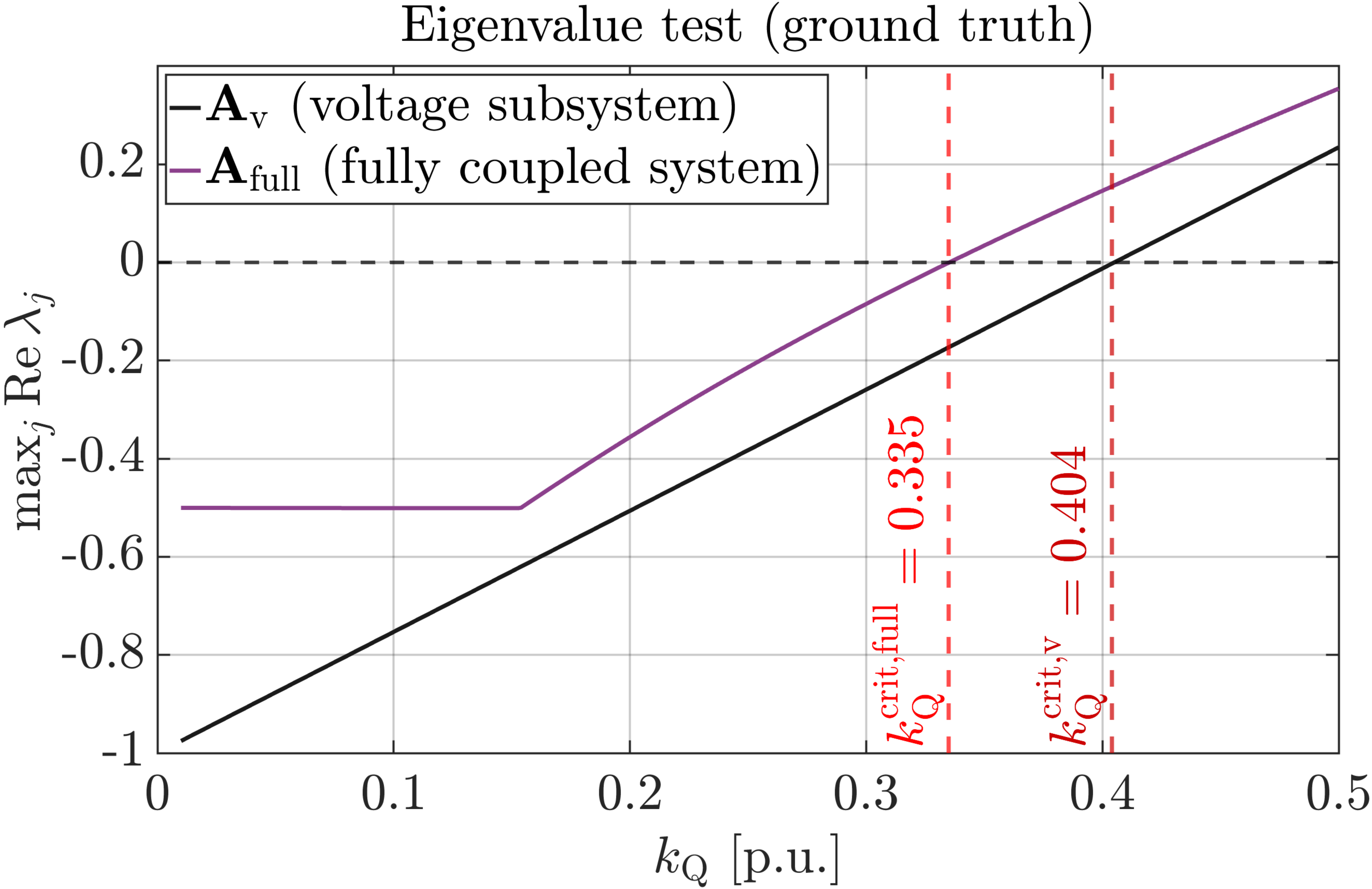}
          \caption{Kron-reduced IEEE 39-bus network: eigenvalue ground truth. Maximum real part of the spectrum vs. the
  homogeneous $k_{\mathrm{Q}}$, for the voltage subsystem
  $\boldsymbol A_{\rm v}$ and for the fully coupled linearization.}
    \label{fig:eigenvalue_ieee39}
\end{figure}

\begin{figure}
    \centering
    \includegraphics[width=0.9\linewidth]{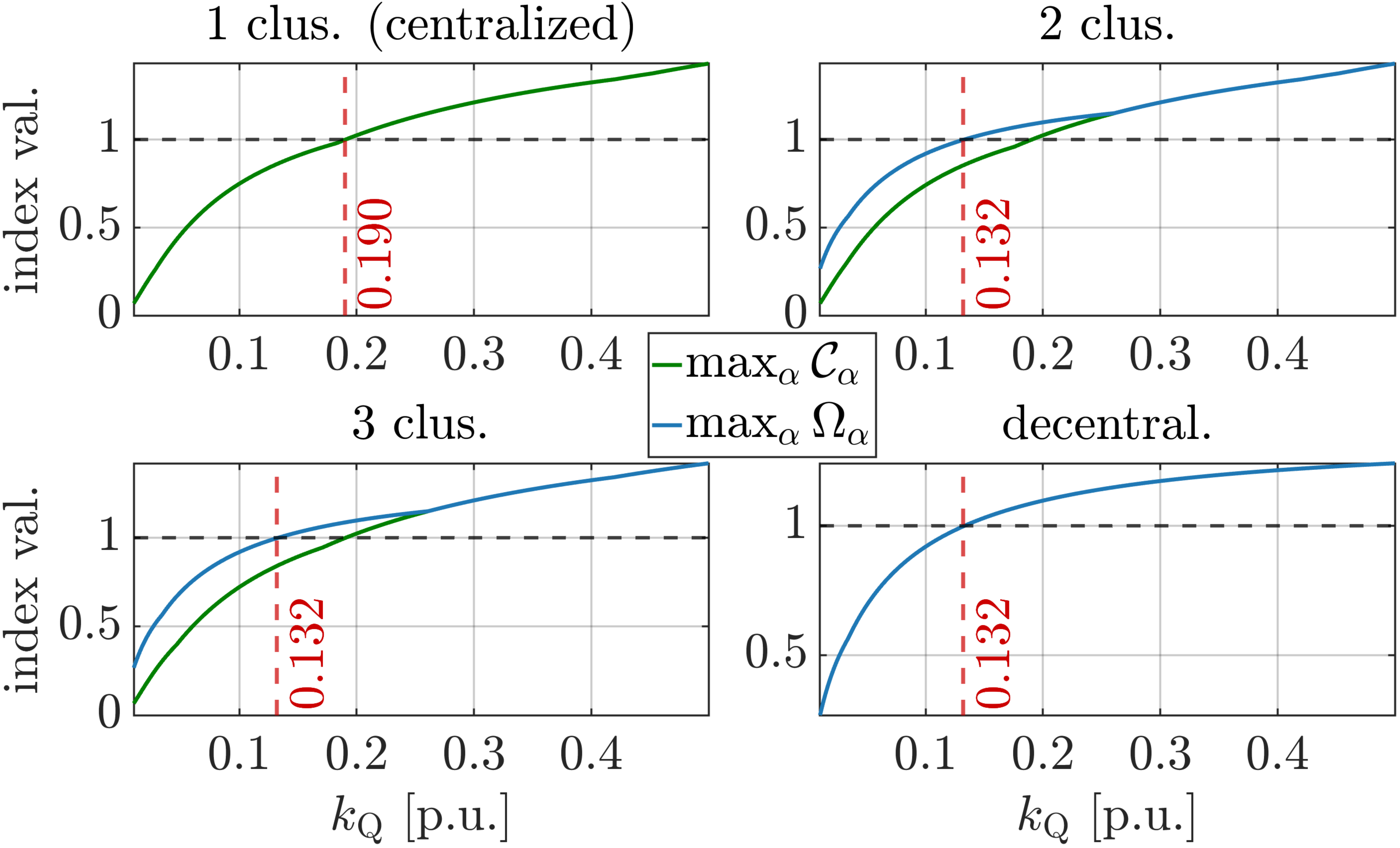}
    \caption{Kron-reduced IEEE 39-bus network, proportional normalization ($x=1$): cluster- and node-level stability indices vs. the homogeneous $k_{\mathrm{Q}}$.}
    \label{fig:indices_KQ_prop_ieee39}
\end{figure}

\begin{figure}
    \centering
    \includegraphics[width=0.9\linewidth]{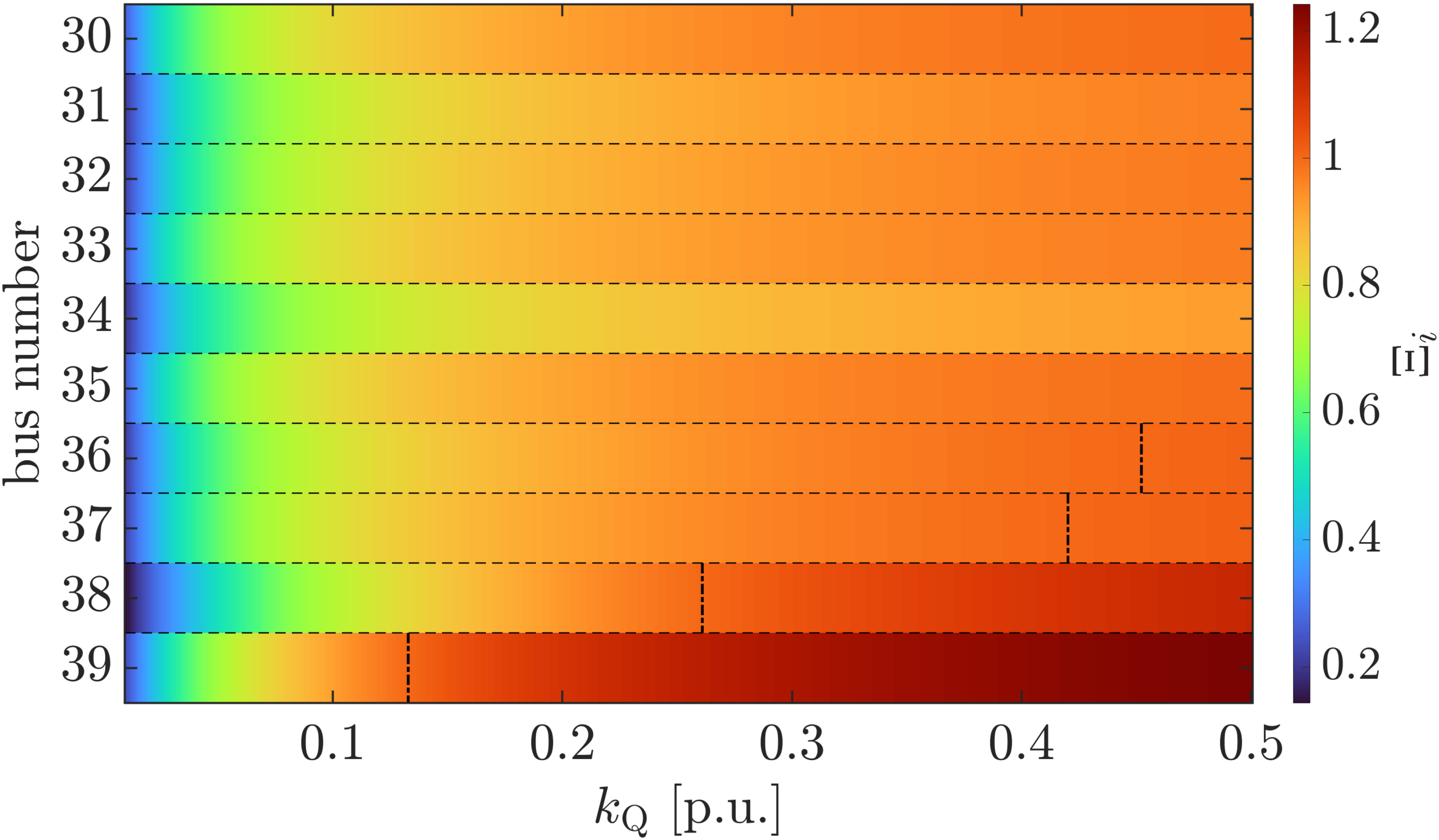}
      \caption{Kron-reduced IEEE 39-bus network, proportional normalization ($x=1$): node-level stability index
  $\Xi_i$ as a function of $k_{\mathrm{Q}}$.}
    \label{fig:heat_map_Xi_prop_ieee39}
\end{figure}

\begin{figure}
    \centering
    \includegraphics[width=0.9\linewidth]{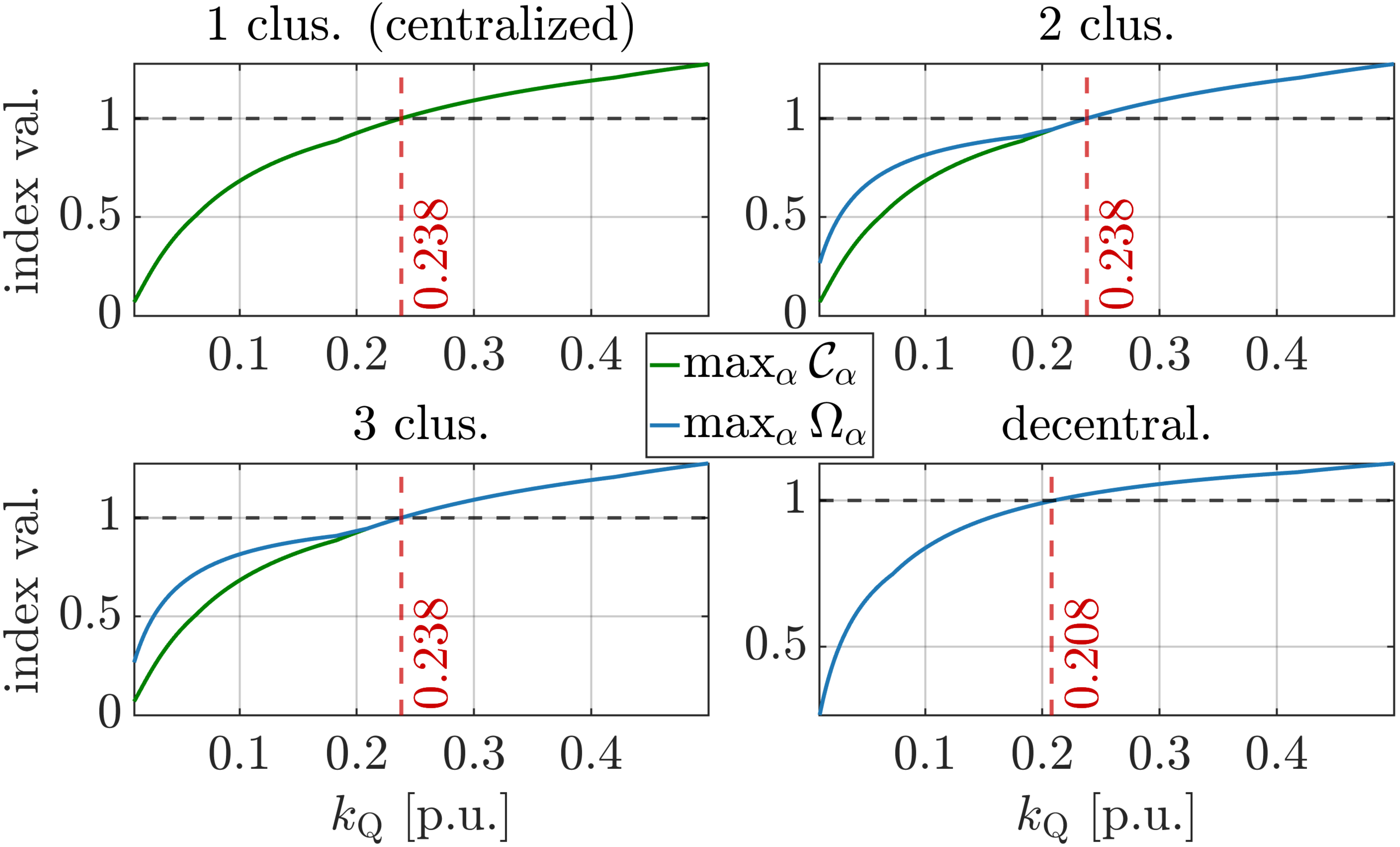}
    \caption{Sparsified Kron-reduced IEEE 39-bus network, proportional normalization ($x=1$): cluster- and node-level indices vs. the homogeneous $k_{\mathrm{Q}}$.}
    \label{fig:indices_KQ_sparsified_prop_ieee39}
\end{figure}

We further evaluate our framework on the New England IEEE
$39$-bus test system. The network is reduced by Kron elimination to its
$n=10$ generator buses (buses $30$-$39$), with bus $39$ taken as the reference. The equilibrium
voltages $\boldsymbol V^{\rm s}$ and the equilibrium angles
$\boldsymbol\theta^{\rm s}$ are obtained from a converged power-flow
solution using MATPOWER \cite{zimmerman2010matpower}. The Kron-reduced graph is dense: each bus is connected to all nine others. In addition to the singleton  and the single-cluster  partition, we consider the 2-cluster partition $\big\{\{30,31,32,33,34\},\{35,36,37,38,39\}\big\}$ and the 3-cluster partition $\big\{\{30,31,32\},\{33,34,35,36\},\{37,38,39\}\big\}$. The filter time constants are set to $\tau_{\mathrm P_i}=\tau_{\mathrm Q_i}=1$s, the active-power/frequency droop gains are set to $k_{\mathrm P_i}=0.05$ p.u., and the homogeneous reactive-power/voltage droop gain $k_{\mathrm Q_i}=k_{\mathrm Q}$ is swept  over $k_{\mathrm Q}\in[0.01,0.5]$ p.u. with a step size of $0.001$ p.u.

We evaluate the same two eigenvalue tests for the decoupled voltage subsystem and the fully coupled linearized system, and the family of certificates used for the synthetic network. The critical $k_{\mathrm Q}$ values are summarized in Table \ref{tab:ieee39_kq_critical_values_x_sweep}. The voltage
subsystem remains Hurwitz up to $k_{\mathrm Q}=0.404$ p.u., whereas the fully coupled linearization loses stability at $k_{\mathrm Q}=0.335$ p.u.; these stability boundaries are also visible in Fig. \ref{fig:eigenvalue_ieee39} as the respective zero crossings of the maximum real parts of the eigenvalues.  

The certificate results again depend strongly on the normalization exponent. Under uniform normalization, $x=0$, all four cyclic small-gain certificates lose validity below $k_{\mathrm Q}=0.01$ p.u. As $x$ increases, all certified ranges expand. Under proportional normalization, $x=1$, the decentralized, 2-cluster, and 3-cluster certificates all fail at $k_{\mathrm Q}=0.132$ p.u., whereas the centralized cyclic certificate remains valid up to $k_{\mathrm Q}=0.190$ p.u. The reason for this coincidence is visible in Figs. \ref{fig:indices_KQ_prop_ieee39} and \ref{fig:heat_map_Xi_prop_ieee39}. The heat map identifies bus $39$ as the first node whose decentralized index $\Xi_i$ reaches unity. Because the Kron-reduced graph is dense, bus $39$ has neighbors outside its cluster under both nontrivial partitions. Under proportional normalization, all incoming gains to bus $39$ are equalized; therefore, when $\Xi_{39}$ reaches unity,
at least one equally large incoming gain is necessarily an inter-cluster channel. The corresponding inter-cluster index must then reach unity as well, which is exactly what occurs for the 2-cluster and 3-cluster partitions in Fig. \ref{fig:indices_KQ_prop_ieee39}. The single-cluster partition is not subject to this limitation because it has no inter-cluster condition, and it therefore remains certified until its internal cycle index reaches unity at $k_{\mathrm Q}=0.190$ p.u. This  behavior is consistent with the mechanism discussed in Section \ref{sec:comparison_partition_normalization} and contrasts
with the synthetic example, where the limiting buses $1,2,5$, and $8$ have all of their neighbors contained within the same cluster under both nontrivial partitions. Their limiting incoming channels are therefore internalized, allowing the 2-cluster and 3-cluster certificates to outperform the decentralized test. Together, the two examples show that the benefit of clustering depends on whether the selected partition internalizes the channels that limit the decentralized certificate.

To demonstrate the same mechanism within the IEEE 39-bus setting, we consider a modified Kron-reduced network in which bus $39$, identified above as the
limiting node, is fully internalized by the nontrivial partitions. We remove the two weak links connecting bus $39$ to buses $34$ and $36$,
and consider the 2-cluster partition
$\big\{\{30,31,32,33,35,37,38,39\},\{34,36\}\big\}$
and the 3-cluster partition
$\big\{\{30,31,32,33,35,37,38,39\},\{34\},\{36\}\big\}$. Under both partitions, bus $39$ and all of its neighbors lie in the same cluster. As shown in Fig. \ref{fig:indices_KQ_sparsified_prop_ieee39}, the 2-cluster and 3-cluster certificates consequently remain valid up to $k_{\mathrm Q}=0.238$ p.u., exceeding the decentralized boundary $k_{\mathrm Q}=0.208$ p.u.

\section{Conclusions}\label{sec:conclusions}

We developed a time-domain framework for small-signal stability certification of lossless grid-forming (GFM) inverter networks with selectable clustering resolution. Starting from a droop-controlled network linearized around a synchronized operating point, a small-angle approximation separates the model into angle-frequency and voltage subsystems. The angle-frequency subsystem is certified by an energy argument using the symmetric weighted-Laplacian structure of the network. For the voltage subsystem, a cyclic small-gain argument yields a family of sufficient stability certificates ranging from the singleton decentralized case to arbitrary cluster partitions and the single-cluster centralized case.

The main structural feature of the voltage certificate is its cluster separability. Expressing voltage-error propagation through
products of node-to-node gains allows the global cyclic small-gain condition to be decomposed into intra-cluster cycle conditions and inter-cluster path conditions. Each cluster therefore verifies its own certificate using intra-cluster data and one-hop boundary information, without assembling the global voltage state matrix. The singleton and single-cluster limits recover the decentralized and centralized cyclic certificates, respectively, while intermediate partitions provide a certification resolution matched to the organization of the network. For any fixed normalization used in the node-to-node gain construction, decentralized certification implies cluster-based certification for every partition, but the converse need not hold.

The resulting indices also expose why a certificate succeeds or fails. The node-level index $\Xi_i$ identifies the limiting incoming channel at an individual inverter, the intra-cluster index $\mathcal C_\alpha$ identifies the most restrictive internal feedback loop, and the inter-cluster index $\Omega_\alpha$ identifies the strongest incoming influence from neighboring clusters. This interpretation also reveals the distinct roles of gain normalization and partitioning. Uniform normalization produces a degree-scaled worst-neighbor effect, whereas proportional normalization equalizes the incoming gains and minimizes the decentralized node-level index at each receiving node over all admissible normalization weights. Under proportional normalization, however, a cluster certificate can improve on the decentralized certificate only if \textit{all} neighbors of the limiting node belong to the same cluster as that node; otherwise, when its decentralized
index reaches unity, a direct inter-cluster path also reaches unity, causing the cluster certificate to fail at the same point. The synthetic and Kron-reduced IEEE 39-bus studies illustrate these mechanisms and, through comparison with eigenvalue-based ground truth, quantify both the performance of the proposed certificates and their remaining conservatism.

\section{Appendix (Proof of Theorem \ref{thm:hierarchical_voltage_stability})}\label{sec:proof}

In this section, we prove Theorem \ref{thm:hierarchical_voltage_stability}. 
Corollary \ref{cor:decentralized_voltage_stability} follows directly from 
Theorem \ref{thm:hierarchical_voltage_stability}, as explained in 
Section \ref{sec:volt_stab_decentralized}, and is therefore not discussed here.    

\subsection{Auxiliary definitions and results}

The following definitions and auxiliary results, adapted from
\cite{karafyllis2011vector}, formalize the MAX-operator framework used in
the cyclic small-gain arguments in the proofs that follow.

\begin{defn}[MAX operator]
For $\boldsymbol u,\boldsymbol y\in\mathbb R^n$, define
\begin{align*}
\operatorname{MAX}\{\boldsymbol u,\boldsymbol y\}
:=
\boldsymbol z,
\end{align*}
where $\boldsymbol z\in\mathbb R^n$ is defined componentwise by 
\begin{align*}
z_i:=\max\{u_i,y_i\},
\qquad i=1,\ldots,n.
\end{align*}
More generally, for $\boldsymbol u_1,\ldots,\boldsymbol u_m\in\mathbb R^n$,
define
\begin{align*}
\operatorname{MAX}\{\boldsymbol u_1,\ldots,\boldsymbol u_m\}
:=
\boldsymbol z,
\end{align*}
where $\boldsymbol z\in\mathbb R^n$ is defined componentwise by
\begin{align*}
z_i:=\max\{u_{1i},\ldots,u_{mi}\},
\qquad i=1,\ldots,n.
\end{align*}
\end{defn}

\begin{defn}[MAX-preserving map]
A mapping $\boldsymbol\Gamma:\mathbb R_+^n\to\mathbb R_+^n$ is called
MAX-preserving if it is nondecreasing and satisfies
\begin{align*}
\boldsymbol\Gamma\bigl(\operatorname{MAX}\{\boldsymbol u,\boldsymbol y\}\bigr)
=
\operatorname{MAX}
\bigl\{
\boldsymbol\Gamma(\boldsymbol u),
\boldsymbol\Gamma(\boldsymbol y)
\bigr\},
\qquad
\forall \boldsymbol u,\boldsymbol y\in\mathbb R_+^n.
\end{align*}
\end{defn}

\begin{lem}[Representation of MAX-preserving maps]
\label{prop:Gamma_representation}
A mapping $\boldsymbol\Gamma:\mathbb R_+^n\to\mathbb R_+^n$, with
\begin{align*}
\boldsymbol\Gamma(\boldsymbol z)
=
\bigl(\Gamma_1(\boldsymbol z),\ldots,\Gamma_n(\boldsymbol z)\bigr)^\top,
\end{align*}
is MAX-preserving if and only if there exist nondecreasing scalar functions
$\gamma_{ij}:\mathbb R_+\to\mathbb R_+$ such that
\begin{align}
\Gamma_i(\boldsymbol z)
=
\max_{j=1,\ldots,n}\gamma_{ij}(z_j),
\qquad
i=1,\ldots,n.
\label{eq:Gamma_representation}
\end{align}
\end{lem}

\begin{lem}[MAX-bound under cyclic small-gain]
\label{lm: bound_with_initial}
Let $\boldsymbol\Gamma:\mathbb R_+^n\to\mathbb R_+^n$ be a continuous
MAX-preserving map with $\boldsymbol\Gamma(\boldsymbol 0_n)=\boldsymbol 0_n$.
Suppose that there exist nondecreasing scalar functions $\gamma_{ij}:\mathbb{R}_+\rightarrow \mathbb{R}_+$ such that
\begin{align*}
\Gamma_i(\boldsymbol z)
=
\max_{j=1,\ldots,n}\gamma_{ij}(z_j),
\qquad
i=1,\ldots,n.
\end{align*}
Assume that the cyclic small-gain conditions hold:
\begin{align*}
\gamma_{ii}(z)<z,
\qquad
\forall z>0,
\qquad
i=1,\ldots,n,
\end{align*}
and, for each $r=2,\ldots,n$,
\begin{align*}
(\gamma_{i_1i_2}\circ\gamma_{i_2i_3}\circ\cdots\circ\gamma_{i_ri_1})(s)<s,
\qquad
\forall s>0,
\end{align*}
for all pairwise distinct indices $i_1,\ldots,i_r\in\{1,\ldots,n\}$, where $\circ$ denotes function
composition.
Define $\boldsymbol\Psi:\mathbb R_+^n\to\mathbb R_+^n$ by
\begin{align*}
\boldsymbol\Psi(\boldsymbol z)
:=
\operatorname{MAX}
\left\{
\boldsymbol z,\,
\boldsymbol\Gamma(\boldsymbol z),\,
\boldsymbol\Gamma^{(2)}(\boldsymbol z),\,
\ldots,\,
\boldsymbol\Gamma^{(n-1)}(\boldsymbol z)
\right\},
\end{align*}
where $\boldsymbol\Gamma^{(r)}$ denotes the $r$-fold composition of $\boldsymbol\Gamma$ with itself. If $\boldsymbol a,\boldsymbol z\in\mathbb R_+^n$ satisfy
\begin{align*}
\boldsymbol z
\leq
\operatorname{MAX}\{\boldsymbol a,\boldsymbol\Gamma(\boldsymbol z)\},
\end{align*}
then
\begin{align*}
\boldsymbol z\leq\boldsymbol\Psi(\boldsymbol a).
\end{align*}
\end{lem}

\subsection{A Centralized Sufficient Criterion for Voltage Stability}

We first derive a centralized cyclic small-gain certificate for the voltage subsystem \eqref{eq:linearized_volt_dynamics}-\eqref{eq:Di_def}.
This serves as the baseline from which the cluster-based stability conditions are obtained.


\begin{prop}[A centralized (network-level) sufficient stability criterion]
\label{thm:centralized_small_gain_certificate} Let $\mathcal V=\{1,\ldots,n\}$ denote the set of grid-forming inverter nodes.
Consider the voltage subsystem
\eqref{eq:linearized_volt_dynamics}-\eqref{eq:Di_def}, obtained under
Assumptions \ref{ass_lossless}-\ref{ass:small_angle_decoupling}. Assume that, for every $i\in\mathcal V$,
$\tau_{\mathrm Q_i}$ and $k_{\mathrm Q_i}$ satisfy
\eqref{eq:kQ_bound_hierarchical}. Let the node-to-node gains
$\gamma_{ik}\ge 0$ be defined as in
Definition \ref{def:hierarchical_voltage_gains}. If every directed cycle satisfies
\begin{align}
\gamma_{j_1j_2}\gamma_{j_2j_3}\cdots\gamma_{j_rj_1}
<
1,
\label{eq:central_cycle_condition}
\end{align}
for all $r \in \{2,\ldots,n\}$ and all pairwise distinct indices $j_1,\ldots,j_r \in \mathcal V$, then the voltage subsystem
\eqref{eq:linearized_volt_dynamics}-\eqref{eq:Di_def} is exponentially stable, and the estimate \eqref{eq:eiss} holds for all $t\ge t_0$.
\end{prop}

\noindent \textbf{Proof.}
For each $i\in\mathcal V$, the variation-of-constants formula applied to
\eqref{eq:linearized_volt_dynamics} gives, for all $t\ge t_0$,


\begin{align*}
\begin{split}
\delta V_i(t)
&=
e^{-\lambda_i(t-t_0)}\delta V_i(t_0)
\\
&\quad+
\frac{k_{\mathrm Q_i}V_i^{\rm s}}{\tau_{\mathrm Q_i}}
\int_{t_0}^t
e^{-\lambda_i(t-\tau)}
\left(
\sum_{k\in\mathcal N_i}
|B_{ik}|\delta V_k(\tau)
\right)d\tau .
\end{split}
\end{align*}


\noindent Since $\lambda_i>0$ for all $i\in\mathcal V$ by
\eqref{eq:kQ_bound_hierarchical}, we obtain, for every $t\ge t_0$ and every
$i\in\mathcal V$,


\begin{align}
\begin{split}
&|\delta V_i(t)|
\le
e^{-\lambda_i(t-t_0)}|\delta V_i(t_0)|
\\
&\quad+
\frac{k_{\mathrm Q_i}V_i^{\rm s}}{\tau_{\mathrm Q_i}}
\sum_{k\in\mathcal N_i}
\int_{t_0}^t
e^{-\lambda_i(t-\tau)}
|B_{ik}||\delta V_k(\tau)|d\tau .
\end{split}
\label{eq:weighted_bound_i}
\end{align}


\noindent Let
\begin{align*}
\lambda:=\min_{i\in\mathcal V}\lambda_i>0.
\end{align*}
Fix a number $\eta\in(0,\lambda)$, to be chosen below. For each
$i\in\mathcal V$ and each $t\ge t_0$, define the weighted supremum
\begin{align}
W_i[t_0,t]
:=
\sup_{\tau\in[t_0,t]}
e^{\eta(\tau-t_0)}|\delta V_i(\tau)|,
\label{eq:weighted_sup_voltage}
\end{align}
and define


\begin{align*}
\boldsymbol W[t_0,t]
:=
\bigl(
W_1[t_0,t],\ldots,W_n[t_0,t]
\bigr)^\top .
\end{align*}


\noindent Multiplying \eqref{eq:weighted_bound_i} by $e^{\eta(t-t_0)}$ yields
\begin{align*}
\begin{split}
&e^{\eta(t-t_0)}|\delta V_i(t)|
\le
e^{-(\lambda_i-\eta)(t-t_0)}|\delta V_i(t_0)|
\\
&+
\frac{k_{\mathrm Q_i}V_i^{\rm s}}{\tau_{\mathrm Q_i}}
\sum_{k\in\mathcal N_i}
|B_{ik}|
\int_{t_0}^t
e^{-(\lambda_i-\eta)(t-\tau)}
e^{\eta(\tau-t_0)}
|\delta V_k(\tau)|d\tau .
\end{split}
\end{align*}
Since $\eta<\lambda\le\lambda_i$, one has $\lambda_i-\eta>0$. Hence, using
\eqref{eq:weighted_sup_voltage},
\begin{align}
\begin{split}
&e^{\eta(t-t_0)}|\delta V_i(t)|
\le
|\delta V_i(t_0)|
\\&\qquad\qquad\qquad+
\frac{k_{\mathrm Q_i}V_i^{\rm s}}{\tau_{\mathrm Q_i}(\lambda_i-\eta)}
\sum_{k\in\mathcal N_i}
|B_{ik}|W_k[t_0,t].
\end{split}
\label{eq:weighted_bound_eta_2}
\end{align}
Now consider some constants $\zeta_{ik}>0, k\in\mathcal{N}_i$ such that $\sum_{k\in\mathcal N_i}\zeta_{ik}=1$. Then,
\begin{align}
\begin{split}
\sum_{k\in\mathcal N_i}|B_{ik}|W_k[t_0,t]
&=\sum_{k\in\mathcal N_i}\zeta_{ik}\frac{|B_{ik}|W_k[t_0,t]}{\zeta_{ik}}\\\leq& \bigg(\sum_{k\in\mathcal N_i}\zeta_{ik}\bigg) \max_{k\in\mathcal N_i}
\frac{|B_{ik}|W_k[t_0,t]}{\zeta_{ik}}\\=&\max_{k\in\mathcal N_i}
\frac{|B_{ik}|W_k[t_0,t]}{\zeta_{ik}}.
\end{split}
\label{eq:sum_to_degree_max_eta}
\end{align}
Recalling $B_{ik}=0$ for all $k\notin\mathcal{N}_i$ and using \eqref{eq:sum_to_degree_max_eta}, we obtain from \eqref{eq:weighted_bound_eta_2} that
\begin{align}
\begin{split}
e^{\eta(t-t_0)}|\delta V_i(t)|
&\!\le\!
|\delta V_i(t_0)|
+
\max_{k\in\mathcal V\setminus\{i\}}
\hat\gamma_{ik,\eta}W_k[t_0,t],
\end{split}
\label{eq:weighted_bound_eta_3}
\end{align}
where
\begin{align*}
\hat\gamma_{ik,\eta}
:=
\begin{cases}
\dfrac{k_{\mathrm Q_i}V_i^{\rm s}|B_{ik}|}
{\tau_{\mathrm Q_i}(\lambda_i-\eta)\zeta_{ik}},
& k\in\mathcal N_i,\\[3mm]
0,
& k\notin\mathcal N_i,
\end{cases}
\end{align*}
or equivalently,
\begin{align*}
\hat\gamma_{ik,\eta}
=
\frac{\lambda_i}{\lambda_i-\eta}\gamma_{ik},
\qquad
i,k\in\mathcal V,
\end{align*}
with $\gamma_{ik}$ defined in \eqref{eq:gammaik_def}.

Applying the estimate
\begin{align*}
a+b
\le
\max\{(1+\varepsilon)a,(1+\varepsilon^{-1})b\},
\qquad
a,b\ge0,\quad \varepsilon>0,
\end{align*}
to \eqref{eq:weighted_bound_eta_3} gives
\begin{align}
\begin{split}
&e^{\eta(t-t_0)}|\delta V_i(t)|
\le
\max
\Biggl\{
(1+\varepsilon)|\delta V_i(t_0)|,
\\
&\qquad\qquad\qquad
\max_{k\in\mathcal V\setminus\{i\}}
(1+\varepsilon^{-1})\hat\gamma_{ik,\eta}W_k[t_0,t]
\Biggr\}.
\end{split}
\label{eq:weighted_max_gain_eta}
\end{align}
Define the inflated gains
\begin{align}
\tilde\gamma_{ik,\eta}
&:=
(1+\varepsilon^{-1})\hat\gamma_{ik,\eta},
\qquad i\neq k,
\label{eq:gamma_tilde_eta_offdiag}
\\
\tilde\gamma_{ii,\eta}
&:=
0,
\qquad i\in\mathcal V,
\label{eq:gamma_tilde_eta_diag}
\end{align}
and define the gain operator
$\boldsymbol\Gamma_\eta:\mathbb R_+^n\to\mathbb R_+^n$ by
\begin{align}
(\Gamma_\eta)_i(\boldsymbol z)
:=
\max_{k\in\mathcal V}\tilde\gamma_{ik,\eta}z_k,
\qquad
i\in\mathcal V.
\label{eq:Gamma_eta_def}
\end{align}
Taking the supremum of \eqref{eq:weighted_max_gain_eta} over $[t_0,t]$ yields
\begin{align*}
W_i[t_0,t]
\le
\max
\left\{
(1+\varepsilon)|\delta V_i(t_0)|,
(\Gamma_\eta)_i\bigl(\boldsymbol W[t_0,t]\bigr)
\right\},
\end{align*}
for all $t\ge t_0$ and $i\in\mathcal V$. Equivalently,
\begin{align}
\boldsymbol W[t_0,t]
\le
\operatorname{MAX}
\left\{
\boldsymbol a(t_0),
\boldsymbol\Gamma_\eta\bigl(\boldsymbol W[t_0,t]\bigr)
\right\},
\label{eq:weighted_sup_gain_eta_vector}
\end{align}
where
\begin{align}
\boldsymbol a(t_0)
:=
\bigl(
a_1(t_0),\ldots,a_n(t_0)
\bigr)^\top,
\quad
a_i(t_0):=(1+\varepsilon)|\delta V_i(t_0)|.
\label{eq:a_t0_eta_def}
\end{align}

We now choose $\eta$ and $\varepsilon$ so that the cyclic small-gain conditions
hold for $\boldsymbol\Gamma_\eta$. By the hypothesis
\eqref{eq:central_cycle_condition}, every directed cycle satisfies
\begin{align*}
\gamma_{j_1j_2}\gamma_{j_2j_3}\cdots\gamma_{j_rj_1}<1,
\end{align*}
for all $r \in \{2,\ldots,n\}$ and all pairwise distinct indices $j_1,\ldots,j_r \in \mathcal V$. For each such fixed directed cycle, the map
\begin{align*}
\eta
\mapsto
\hat\gamma_{j_1j_2,\eta}\,
\hat\gamma_{j_2j_3,\eta}\,
\cdots\,
\hat\gamma_{j_rj_1,\eta}
\end{align*}
is continuous on $[0,\lambda)$ and satisfies
\begin{align*}
\lim_{\eta\to0^+}
\hat\gamma_{j_1j_2,\eta}\,
\hat\gamma_{j_2j_3,\eta}\,
\cdots\,
\hat\gamma_{j_rj_1,\eta}\,
=
\gamma_{j_1j_2}
\gamma_{j_2j_3}
\cdots
\gamma_{j_rj_1}.
\end{align*}
Since the cycle inequalities in \eqref{eq:central_cycle_condition} are strict
and since there are only finitely many directed cycles with pairwise distinct
vertices, there exists $\eta\in(0,\lambda)$ sufficiently small such that
\begin{align*}
\sigma_\eta
:=
\max_{2\le r\le n}
\max_{\substack{j_1,\ldots,j_r\in\mathcal V\\
\text{pairwise distinct}}}
\hat\gamma_{j_1j_2,\eta}\,
\hat\gamma_{j_2j_3,\eta}\,
\cdots\,
\hat\gamma_{j_rj_1,\eta}
<1.
\end{align*}

For any such cycle, by \eqref{eq:gamma_tilde_eta_offdiag},
\begin{align*}
\begin{split}
&\tilde\gamma_{j_1j_2,\eta}\,
\tilde\gamma_{j_2j_3,\eta}\,
\cdots\,
\tilde\gamma_{j_rj_1,\eta}
\\
&\qquad =
(1+\varepsilon^{-1})^r\,
\hat\gamma_{j_1j_2,\eta}\,
\hat\gamma_{j_2j_3,\eta}\,
\cdots\,
\hat\gamma_{j_rj_1,\eta}.
\end{split}
\end{align*}
Thus
\begin{align}
\tilde\gamma_{j_1j_2,\eta}
\tilde\gamma_{j_2j_3,\eta}
\cdots
\tilde\gamma_{j_rj_1,\eta}
\le
(1+\varepsilon^{-1})^r\sigma_\eta.
\label{eq:tilde_cycle_bound_eta}
\end{align}
If $\sigma_\eta=0$, any $\varepsilon>0$ is admissible. If $\sigma_\eta>0$,
choose $\varepsilon>0$ sufficiently large so that
\begin{align}
(1+\varepsilon^{-1})^n\sigma_\eta<1.
\label{eq:epsilon_eta_choice}
\end{align}
For example, it is sufficient to choose $\varepsilon
>
\frac{1}
{\left(\frac{1}{\sigma_\eta}\right)^{\frac{1}{n}}-1}.$ Since $r\in\{2,\ldots,n\}$, \eqref{eq:tilde_cycle_bound_eta} and
\eqref{eq:epsilon_eta_choice} imply
\begin{align}
\tilde\gamma_{j_1j_2,\eta}\,
\tilde\gamma_{j_2j_3,\eta}\,
\cdots\,
\tilde\gamma_{j_rj_1,\eta}
<1,
\label{eq:tilde_cycle_eta_condition}
\end{align}
for all $r\in\{2,\ldots,n\}$ and all pairwise distinct
$j_1,\ldots,j_r\in\mathcal V$. Moreover, by
\eqref{eq:gamma_tilde_eta_diag},
\begin{align}
\tilde\gamma_{ii,\eta}=0<1,
\qquad
i\in\mathcal V.
\label{eq:self_gain_eta_condition}
\end{align}

Because $\boldsymbol\Gamma_\eta$ has the form \eqref{eq:Gamma_eta_def} with
nonnegative scalar gains, it is MAX-preserving (see
Lemma \ref{prop:Gamma_representation}), and satisfies
$\boldsymbol\Gamma_\eta(\boldsymbol 0_n)=\boldsymbol 0_n$. Moreover, \eqref{eq:tilde_cycle_eta_condition}
and \eqref{eq:self_gain_eta_condition} are precisely the cyclic small-gain
conditions for $\boldsymbol\Gamma_\eta$. Therefore, by
Lemma \ref{lm: bound_with_initial}, applied to
\eqref{eq:weighted_sup_gain_eta_vector}, we obtain
\begin{align}
\boldsymbol W[t_0,t]
\le
\boldsymbol\Psi_\eta\bigl(\boldsymbol a(t_0)\bigr),
\qquad
t\ge t_0,
\label{eq:W_Psi_eta_bound}
\end{align}
where
\begin{align}
\boldsymbol\Psi_\eta(\boldsymbol z)
:=
\operatorname{MAX}
\left\{
\boldsymbol z,
\boldsymbol\Gamma_\eta(\boldsymbol z),
\boldsymbol\Gamma_\eta^{(2)}(\boldsymbol z),
\ldots,
\boldsymbol\Gamma_\eta^{(n-1)}(\boldsymbol z)
\right\}.
\label{eq:Psi_eta_def}
\end{align}

Since
\begin{align*}
e^{\eta(t-t_0)}|\delta V_i(t)|
\le
W_i[t_0,t],
\end{align*}
it follows from \eqref{eq:W_Psi_eta_bound} that, for all $t\ge t_0$ and
$i\in\mathcal V$,
\begin{align}
|\delta V_i(t)|
\le
e^{-\eta(t-t_0)}
\left(
\boldsymbol\Psi_\eta\bigl(\boldsymbol a(t_0)\bigr)
\right)_i.
\label{eq:component_decay_eta_1}
\end{align}
Define the nonnegative vector
\begin{align*}
\boldsymbol w(t_0)
:=
\bigl(
|\delta V_1(t_0)|,\ldots,|\delta V_n(t_0)|
\bigr)^\top.
\end{align*}
Then, by \eqref{eq:a_t0_eta_def},
\begin{align*}
\boldsymbol a(t_0)
=
(1+\varepsilon)\boldsymbol w(t_0).
\end{align*}
Thus \eqref{eq:component_decay_eta_1} gives
\begin{align}
|\delta V_i(t)|
\le
e^{-\eta(t-t_0)}
\left(
\boldsymbol\Psi_\eta\bigl((1+\varepsilon)\boldsymbol w(t_0)\bigr)
\right)_i.
\label{eq:component_decay_eta_2}
\end{align}

We next use monotonicity and positive homogeneity of $\boldsymbol\Psi_\eta$.
First, since
\begin{align*}
\boldsymbol w(t_0)
\le
\|\boldsymbol w(t_0)\|_2\boldsymbol 1_n,
\end{align*}
and since $\boldsymbol\Psi_\eta$ is monotone, \eqref{eq:component_decay_eta_2}
implies
\begin{align}
|\delta V_i(t)|
\le
e^{-\eta(t-t_0)}
\left(
\boldsymbol\Psi_\eta\bigl(
(1+\varepsilon)\|\boldsymbol w(t_0)\|_2\boldsymbol 1_n
\bigr)
\right)_i.
\label{eq:component_decay_eta_3}
\end{align}
Second, $\boldsymbol\Gamma_\eta$ is positively homogeneous. Indeed, for every
$c\ge0$, every $\boldsymbol z\in\mathbb R_+^n$, and every $i\in\mathcal V$,
\begin{align*}
(\Gamma_\eta)_i(c\boldsymbol z)
&=
\max_{k\in\mathcal V}\tilde\gamma_{ik,\eta}(cz_k) =
c\max_{k\in\mathcal V}\tilde\gamma_{ik,\eta}z_k
=
c(\Gamma_\eta)_i(\boldsymbol z).
\end{align*}
Hence
\begin{align*}
\boldsymbol\Gamma_\eta(c\boldsymbol z)
=
c\boldsymbol\Gamma_\eta(\boldsymbol z),
\qquad
\forall c\ge0,\ \boldsymbol z\in\mathbb R_+^n.
\end{align*}
By induction,
\begin{align*}
\boldsymbol\Gamma_\eta^{(m)}(c\boldsymbol z)
=
c\boldsymbol\Gamma_\eta^{(m)}(\boldsymbol z),
\qquad
\forall c\ge0,\ \boldsymbol z\in\mathbb R_+^n,\ m\in\mathbb N.
\end{align*}
Therefore, using \eqref{eq:Psi_eta_def},
\begin{align}
&\boldsymbol\Psi_\eta(c\boldsymbol z)
=
\operatorname{MAX}
\left\{
c\boldsymbol z,
c\boldsymbol\Gamma_\eta(\boldsymbol z),
\ldots,
c\boldsymbol\Gamma_\eta^{(n-1)}(\boldsymbol z)
\right\}
\nonumber\\
&=
c
\operatorname{MAX}
\left\{
\boldsymbol z,
\boldsymbol\Gamma_\eta(\boldsymbol z),
\ldots,
\boldsymbol\Gamma_\eta^{(n-1)}(\boldsymbol z)
\right\} =
c\boldsymbol\Psi_\eta(\boldsymbol z).
\label{eq:Psi_eta_homogeneous}
\end{align}
Applying \eqref{eq:Psi_eta_homogeneous} to
\eqref{eq:component_decay_eta_3} gives
\begin{align*}
|\delta V_i(t)|
\le
(1+\varepsilon)e^{-\eta(t-t_0)}
\left(
\boldsymbol\Psi_\eta(\boldsymbol 1_n)
\right)_i
\|\boldsymbol w(t_0)\|_2.
\end{align*}
Squaring, summing over $i\in\mathcal V$, and taking square roots
\begin{align*}
\|\delta\boldsymbol V(t)\|_2
\le
(1+\varepsilon)
\|\boldsymbol\Psi_\eta(\boldsymbol 1_n)\|_2
e^{-\eta(t-t_0)}
\|\boldsymbol w(t_0)\|_2.
\end{align*}
Finally, since $\|\boldsymbol w(t_0)\|_2
=
\|\delta\boldsymbol V(t_0)\|_2$, we conclude that
\begin{align*}
\|\delta\boldsymbol V(t)\|_2
\le
M_{\rm v}e^{-\lambda_{\rm v}(t-t_0)}
\|\delta\boldsymbol V(t_0)\|_2,
\qquad
t\ge t_0,
\end{align*}
where
\begin{align*}
M_{\rm v}
:=
(1+\varepsilon)
\|\boldsymbol\Psi_\eta(\boldsymbol 1_n)\|_2
>0,
\qquad
\lambda_{\rm v}:=\eta>0.
\end{align*}
This completes the proof of Proposition
\ref{thm:centralized_small_gain_certificate}.
\hfill $\square$

\subsection{Derivation of the Cluster-Based Voltage Stability Conditions (Theorem \ref{thm:hierarchical_voltage_stability})}
It is enough to verify the cyclic small-gain condition
\eqref{eq:central_cycle_condition} in
Proposition \ref{thm:centralized_small_gain_certificate}. Once this is shown, the
estimate \eqref{eq:eiss} follows directly from Proposition \ref{thm:centralized_small_gain_certificate}.

Let
\begin{align}
(j_1,j_2,\ldots,j_r,j_1),
\label{eq:cluster_proof_node_cycle}
\end{align}
be an arbitrary directed cycle of pairwise distinct nodes, where
$r\in\{2,\ldots,n\}$. We use the cyclic convention $j_{r+1}:=j_1.$ We must prove that
\begin{align}
\prod_{\nu=1}^{r}\gamma_{j_\nu j_{\nu+1}}<1.
\label{eq:cluster_proof_goal}
\end{align}

For each $\nu\in\{1,\ldots,r\}$, let $\alpha_\nu$ be the unique cluster index such
that $
j_\nu\in\mathcal A_{\alpha_\nu}.$ There are two cases.

\medskip
\noindent\textit{Case 1: the node cycle is contained entirely in one cluster.}

Suppose that there exists $\alpha\in\{1,\ldots,m\}$ such that $j_1,\ldots,j_r\in\mathcal A_\alpha.$ If condition (i) holds for $\mathcal A_\alpha$, then
\eqref{eq:cluster_proof_goal} follows immediately.

\medskip
\noindent\textit{Case 2: the cycle visits at least two distinct clusters.}

Since the cycle visits at least two distinct clusters, the cluster sequence
$\alpha_1,\ldots,\alpha_r$ is not constant. By cyclically relabeling the node cycle, choose $j_1$ to be a node immediately following a transition between two distinct clusters. 

We now partition the node cycle into maximal consecutive segments that remain in
a single cluster and leave that cluster only at the last edge of the segment. More
precisely, there exist integers $
1=p_1<p_2<\cdots<p_q\le r,\; q\ge 2$ such that, if we set $p_{q+1}:=r+1,$
then for each $\ell\in\{1,\ldots,q\}$ the nodes $j_{p_\ell},\,j_{p_\ell+1},\,\ldots,\,j_{p_{\ell+1}-1}$
belong to the same cluster, while the next node $j_{p_{\ell+1}}$ belongs to a
different cluster, where $j_{r+1}:=j_1$. Let $\beta_\ell$ denote the cluster index of
the $\ell$th segment, that is, $
j_{p_\ell},\,j_{p_\ell+1},\,\ldots,\,j_{p_{\ell+1}-1}
\in
\mathcal A_{\beta_\ell},$
and $j_{p_{\ell+1}}\in\mathcal A_{\beta_{\ell+1}},
\;
\beta_{\ell+1}\neq\beta_\ell$
for $\ell=1,\ldots,q$, where we adopt the cyclic convention $
\beta_{q+1}:=\beta_1.$ Thus the $\ell$th segment is the directed path
\begin{align}
\bigl(
j_{p_\ell},
j_{p_\ell+1},
\ldots,
j_{p_{\ell+1}}
\bigr),
\label{eq:cluster_proof_lth_segment_path}
\end{align}
which starts in cluster $\mathcal A_{\beta_\ell}$, stays inside
$\mathcal A_{\beta_\ell}$ up to the node $j_{p_{\ell+1}-1}$, and enters the
next cluster $\mathcal A_{\beta_{\ell+1}}$ at the final step.

For each $\ell\in\{1,\ldots,q\}$, define the gain of the $\ell$th segment by
\begin{align*}
\Pi_\ell
:=
\prod_{\nu=p_\ell}^{p_{\ell+1}-1}
\gamma_{j_\nu j_{\nu+1}},
\end{align*}
where $j_{r+1}:=j_1$. Since the original cycle
\eqref{eq:cluster_proof_node_cycle} is simple, the nodes in each segment are
pairwise distinct. Therefore the path
\eqref{eq:cluster_proof_lth_segment_path} is admissible in the definition
\eqref{eq:cluster_gain_def} of the cluster-to-cluster gain
$\bar\gamma_{\beta_\ell\beta_{\ell+1}}$, and hence $\Pi_\ell
\le
\bar\gamma_{\beta_\ell\beta_{\ell+1}},
\,
\ell=1,\ldots,q.$
Multiplying these inequalities yields
\begin{align}
\gamma_{j_1j_2}\gamma_{j_2j_3}\cdots\gamma_{j_rj_1}
=
\prod_{\ell=1}^{q}\Pi_\ell
\le
\prod_{\ell=1}^{q}
\bar\gamma_{\beta_\ell\beta_{\ell+1}}.
\label{eq:cluster_proof_node_product_cluster_product}
\end{align}
If $\bar\gamma_{\beta_\ell\beta_{\ell+1}}=0$ for some $\ell\in\{1,\ldots,q\}$, then
\eqref{eq:cluster_proof_node_product_cluster_product} immediately gives $\prod_{\ell=1}^{q}
\bar\gamma_{\beta_\ell\beta_{\ell+1}}
=0,$
and hence \eqref{eq:cluster_proof_goal} holds. Otherwise, $\bar\gamma_{\beta_\ell\beta_{\ell+1}}>0,\,
\ell=1,\ldots,q.$ By the definition of the effective inter-cluster neighbor set \eqref{eq:cluster_neighbor_set_def}, this implies
$\beta_{\ell+1}
\in
\mathcal N_{\beta_\ell}^{\rm cluster},
\,
\ell=1,\ldots,q.$ Thus, condition \eqref{eq:inter_cluster_gain_condition} yields $
\bar\gamma_{\beta_\ell\beta_{\ell+1}}<1,\,
\ell=1,\ldots,q.$ Since there are finitely many nonnegative factors and each is strictly smaller than
one, we have
$\prod_{\ell=1}^{q}
\bar\gamma_{\beta_\ell\beta_{\ell+1}}
<1.$ Combining this inequality with
\eqref{eq:cluster_proof_node_product_cluster_product} gives $\gamma_{j_1j_2}\gamma_{j_2j_3}\cdots\gamma_{j_rj_1}
<1.$ Thus, \eqref{eq:cluster_proof_goal}  holds.

Since the directed cycle \eqref{eq:cluster_proof_node_cycle} was arbitrary, every
directed cycle of pairwise distinct nodes in $\mathcal V$ satisfies
\eqref{eq:central_cycle_condition}. Hence, the conclusion follows directly from
Proposition \ref{thm:centralized_small_gain_certificate}. \hfill $\square$

\bibliographystyle{IEEEtran}
\bibliography{main}
\end{document}